\documentclass[review]{elsarticle}

\usepackage{lineno,hyperref}
\usepackage{amsmath}
\usepackage{amssymb}
\usepackage{amsthm}
\usepackage{url}
\usepackage[utf8]{inputenc}
\usepackage{stmaryrd}
\usepackage{tikz}
\usepackage{array}
\usepackage{caption}
\usepackage{floatrow}
\usepackage{tabularx}
\usepackage{verbatim}
\usepackage{algorithm}
\usepackage[noend]{algpseudocode}
\usepackage{enumitem}
\usepackage{subcaption}
\usepackage{soul}
\usepackage{ulem}
\usepackage{cancel}

\makeatletter
\newcommand{\labitem}[2]{%
\def\@itemlabel{\textbf{#1}}
\item
\def\@currentlabel{#1}\label{#2}}
\makeatother

\newcommand{\throw}{\Uparrow}
\newcommand{\annotate}{annotate}
\newcommand{\labexp}{labexp}
\newcommand{\addlab}{addlab}
\newcommand{\calck}{calck}
\newcommand{\muni}{matchUni}
\newcommand{\lexrule}{islexrule}
\newcommand{\isunique}{isUniLex}
\newcommand{\fst}{FIRST}
\newcommand{\flw}{FOLLOW}
\newcommand{\notannot}{notlabel}
\newcommand{\ban}{ban}
\newcommand*\Let[2]{\State #1 $\gets$ #2}
\newcommand{\Epsi}{\varepsilon}
\newcommand{\Rp}{R^\prime}
\newcommand{\tm}[1]{\textrm{`{\tt #1}'}}
\newcommand{\lab}[2]{[#1]^{\tt #2}}
\newcommand{\labs}[2]{[#1]^{\cancel{\tt #2}}}
\newcommand{\labt}[2]{[{\tt #1}]^{\tt #2}}
\newcommand{\labst}[2]{[{\tt #1}]^{\cancel{\tt #2}}}
\newcommand{\inp}[1]{``\texttt{#1}"}

\newcommand{\fivespaces}{\;\;\;\;\;}
\newcommand{\tenspaces}{\fivespaces\fivespaces}

\newcommand{\mylabel}[1]{\, \mathbf{(#1)}}
\newcommand{\Lp}{\stackrel{\mbox{\tiny{PEG}}}{\leadsto}}

\newcommand{\Tup}[2]{(#1,\,#2)}
\newcommand{\Xp}{x^\prime}
\newcommand{\Xpp}{x^{\prime\prime}}
\newcommand{\Yp}{y^\prime}

\newcommand{\Wp}{w^\prime}

\newcommand{\Lfail}{{\tt fail}}

\newcommand{\Peg}[2]{#1[#2]}
\newcommand{\Pgg}[1]{\Peg{G}{#1}}

\newcommand{\Matk}[3]{#1\;\,#2\;\,#3\,}
\newcommand{\Matgk}[3]{\Matk{\Pgg{#1}}{#3}{#2}}

\newcommand{\Vlex}{V_{Lex}}
\newcommand{\Vsyn}{V_{Syn}}

\newtheorem{definition}{Definition}
\newtheorem{lemma}{Lemma}

\newenvironment{varalgorithm}[1]
  {\algorithm\renewcommand{\thealgorithm}{#1}}
  {\endalgorithm}

\modulolinenumbers[5]

\journal{Science of Computer Programming}

\begin{document}

\begin{frontmatter}

\author{Sérgio Queiroz de Medeiros}
\ead{sergiomedeiros@ect.ufrn.br}
\address{School of Science and Technology -- UFRN --
Natal -- Brazil}

\author{Gilney de Azevedo Alvez Junior}
\ead{gilneyjnr@gmail.com}
\address{Institute Digital Metropolis -- UFRN --
Natal -- Brazil}

\author{Fabio Mascarenhas}
\ead{mascarenhas@ufrj.br}
\address{Department of Computer Science -- UFRJ --
Rio de Janeiro -- Brazil}

\title{Automatic Syntax Error Reporting and Recovery in Parsing Expression Grammars}

\begin{abstract}
Error recovery is an essential feature for a parser that
should be plugged in Integrated Development Environments (IDEs),
which must build Abstract Syntax Trees (ASTs) even for syntactically
invalid programs in order to offer features such as automated refactoring
and code completion.

Parsing Expressions Grammars (PEGs) are a formalism that
naturally describes recursive top-down parsers using a restricted form
of backtracking. Labeled failures are a
conservative extension of PEGs that adds an error reporting
mechanism for PEG parsers, and these labels can also be associated
with recovery expressions to provide an error recovery mechanism.
These expressions can use the full expressivity of PEGs to recover
from syntactic errors.

Manually annotating a large grammar with labels and recovery expressions
can be difficult. In this work, we present two approaches, \emph{Standard}
and \emph{Unique}, to automatically annotate a PEG with labels, and to build
their corresponding recovery expressions. The Standard approach annotates a
grammar in a way similar to manual annotation, but it may insert labels
incorrectly, while the Unique approach is more conservative to annotate
a grammar and does not insert labels incorrectly. 

We evaluate both approaches by using them to generate error
recovering parsers for four programming languages: Titan, C, Pascal and Java.
In our evaluation, the parsers produced using the Standard approach,
after a manual intervention to remove the labels incorrectly added, gave
an acceptable recovery for at least $70\%$ of the files in each language.
By it turn, the acceptable recovery rate of the parsers produced via the Unique
approach, without the need of manual intervention, ranged from $41\%$ to $76\%$.
\end{abstract}

\begin{keyword}
parsing expression grammars \sep
labeled failures \sep 
error reporting \sep
error recovery
\end{keyword}

\end{frontmatter}


\section{Introduction}
\label{sec:intro}

Integrated Development Environments (IDEs) often require parsers
that can recover from syntax errors and build syntax trees
even for syntactically invalid programs, in other to conduct
further analyses necessary for IDE features such as automated refactoring
and code completion.

Parsing Expression Grammars (PEGs)~\cite{ford2004peg} are a formalism
used to describe the syntax of programming languages, as an alternative
for Context-Free Grammars (CFGs).
We can view a PEG as a formal description of a recursive top-down parser
for the language it describes. PEGs have a concrete syntax based on
the syntax of {\em regexes}, or extended regular expressions.
Unlike CFGs,
PEGs avoid ambiguities in the definition of the grammar's
language by construction, due to the use of an {\em ordered choice} operator.

The ordered choice operator naturally maps to restricted (or local)
backtracking in a recursive top-down parser.
The alternatives of a choice are tried in
order; when the first alternative recognizes an input prefix,
no other alternative of this choice is tried, but when an
alternative fails to recognize an input prefix, the parser
backtracks to the same input position it was before
trying this alternative and then tries the next one.

A naive interpretation of PEGs is problematic when dealing with
inputs with syntactic errors, as a failure during parsing an input
is not necessarily an error, but can be just an indication that the
parser should backtrack and try another alternative. 
Labeled failures~\cite{maidl2013peglabel,maidl2016peglabel} are
a conservative extension of PEGs that address this problem of
error reporting in PEGs by using explicit error labels, which are distinct
from a regular failure. We throw a label to signal an error
during parsing, and each label can then be tied to a
specific error message.

We can leverage the same labels to add an error recovery mechanism,
by attaching a recovery expression to each label. This expression
is just a regular parsing expression, and it usually skips the erroneous
input until reaching a synchronization point, while producing a
dummy AST node~\cite{medeiros2018sac,medeiros2018visual}.

Labeled failures produce good error messages and error recovery,
but they can add a considerable annotation burden in large grammars,
as each point where we want to signal and recover from a syntactic error
must be explicitly marked.

In a previous work~\cite{medeiros2018sblp}, we presented
the Algorithm~\ref{alg:standard}, which automatically annotates
a PEG with labels and builds their corresponding recovery expressions.
We evaluated the use of such algorithm to build an error recovering
parser for the Titan programming language. 

This paper extends the previous one by also evaluating the use
of Algorithm~\ref{alg:standard} to build error recovering parsers
for C, Pascal and Java.

As pointed out in~\cite{medeiros2018sblp}, 
Algorithm~\ref{alg:standard} may add some labels incorrectly,
which would prevent the parser from recognizing
syntactically valid programs.

In this paper we try to address this issue by proposing
the Algorithm~\ref{alg:unique}, which inserts labels in
a more conservative way. The use of Algorithm~\ref{alg:unique}
avoids the problem of adding labels incorrectly,
although it inserts less labels than Algorithm~\ref{alg:standard}.

Overall, our experiments show that Algorithm~\ref{alg:standard}
can be used to produce error recovering parsers with
the help of manual intervention, which was small in
case of our Titan, C, and Pascal grammars, and more
significant in case of Java. By its turn, 
Algorithm~\ref{alg:unique} can be used to
automatically generate functional error
recovering parsers, whose error recovery
quality is lower when compared to the
parsers got via Algorithm~\ref{alg:standard}.

The remainder of this paper is organized as follows:
Section~\ref{sec:pegs} discusses error recovery in PEGs
using labeled failures and recovery expressions;
Section~\ref{sec:algostandard} shows Algorithm~\ref{alg:standard},
which automatically annotates a PEG with labels and associates
a recovery expression to each label;
Section~\ref{sec:evalstandard} evaluates the use of
Algorithm~\ref{alg:standard} to annotate the grammars of
four programming languages: Titan, C,
Pascal, and Java; Section~\ref{sec:algocon} discusses
conservative approaches to insert labels and presents 
Algorithm~\ref{alg:unique}, which inserts only correct
labels; Section~\ref{sec:evalcon} compares the use of
both algorithms to annotate Titan, C, Pascal and Java grammars;
Section~\ref{sec:rel} discusses related work on error
reporting and error recovery;
finally, Section~\ref{sec:conc} gives some concluding remarks.

\section{Error Recovery in PEGs with Labeled Failures}
\label{sec:pegs}

In this section we present an introduction to labeled
PEGs and discuss how to build an error recovery mechanism for PEGs
by attaching a recovery expression to each labeled failure.

A labeled PEG $G$ is a tuple $(V,T,P,L,R,\Lfail, p_{S})$, where
$V$ is a finite set of non-terminals, $T$ is a finite set of terminals,
$P$ is a total function from non-terminals to parsing expressions,
$L$ is a finite set of labels, $R$ is a function from labels to
parsing expressions, $\Lfail \notin L$ is a failure label,
and $p_{S}$ is the initial parsing expression.
We will use the term {\em recovery expression} when referring to
the parsing expression associated with a given label.
We will assume that $V = \Vlex \cup \Vsyn$, where $\Vlex$
is the set of non-terminals that match lexical elements,
also known as tokens, and $\Vsyn$ represents the non-terminals
that match syntactical elements. When describing the PEG
for a given language, we will use names in uppercase for
the lexical non-terminals. From now on, unless otherwise noted,
we will use PEG as synonymous to labeled PEG.

We describe the function $P$ as a set of rules of the form
$A \leftarrow p$, where $A \in V$ and $p$ is a parsing
expression. A parsing expression $p$, when applied to an
input string $s$, either succeeds or fails.
When the matching of $p$ succeeds, it consumes a prefix
of the input and returns the remaining suffix,
and when it fails, it produces a label, associated
with an input suffix.

The abstract syntax of parsing expressions is as follows,
where $p$, $p_1$ and $p_2$ are parsing expressions:
$\Epsi$ represents the empty string, $a \in T$ denotes
a terminal, $A \in V$ represents a non-terminal,
$p_1 p_2$ is a concatenation, $p_1 \;/\; p_2$ is an
ordered choice, $p*$ indicates zero or more repetitions,
$!p$ is a negative predicate, and $\throw^l$ throws a
label $l \in L$. 

\begin{figure}[t!]
	{\scriptsize
	\begin{align*}
		& \textbf{Empty} \fivespaces
		{\frac{}{G[\varepsilon] \; R \; x \Lp x}} \mylabel{empty.1}
		\fivespaces
		\textbf{Non-terminal} \tenspaces
		{\frac{G[P(A)] \; R \; xy \Lp X}
			{G[A] \; R \; xy \Lp X}} \mylabel{var.1}
		\\ \\
		& \textbf{Terminal} \fivespaces
		{\frac{}{G[a] \; R \; ax \Lp x}} \mylabel{term.1} 
		\fivespaces
		{\frac{b \ne a}{G[b] \; R \; ax \Lp \Tup{\Lfail}{ax}}} \mylabel{term.2}
		\fivespaces
		{\frac{}{G[a] \; R \; \varepsilon \Lp \Tup{\Lfail}{\varepsilon}}} \mylabel{term.3}
		\\ \\
		& \textbf{Sequence} \fivespaces
		{\frac{G[p_1] \; R \; xy \Lp y   \fivespaces  G[p_2] \; R \; y \Lp X}
			{G[p_1 \; p_2] \; R \; xy \Lp X}}   \mylabel{seq.1}
        \fivespaces
		{\frac{G[p_1] \; R \; xy \Lp \Tup{f}{y}}
			{G[p_1 \; p_2] \; R \; xy \Lp \Tup{f}{y}}}   \mylabel{seq.2}
		\\ \\
		& \textbf{Ordered Choice} \fivespaces
		{\frac{\Matgk{p_1}{xy}{R} \Lp y}
        {\Matgk{p_1 \;\slash\; p_2}{xy}{R} \Lp y}} \mylabel{ord.1}
		\fivespaces
		{\frac{\Matgk{p_1}{xy}{R} \Lp \Tup{l}{y} \fivespaces l \neq \Lfail}
        {\Matgk{p_1 \;\slash\; p_2}{xy}{R} \Lp \Tup{l}{y}}} \mylabel{ord.2}
      \\ \\
		& \tenspaces
		{\frac{\Matgk{p_1}{xy}{R} \Lp \Tup{\Lfail}{y} \fivespaces \Matgk{p_2}{xy}{R} \Lp X}
      {\Matgk{p_1 \;\slash\; p_2}{xy}{R} \Lp X}} \mylabel{ord.3}
		\\ \\
		& \textbf{Repetition} \fivespaces
		{\frac{G[p] \; R \; xy \Lp \Tup{\Lfail}{y}}
			{G[p*] \; R \; xy \Lp xy}} \mylabel{rep.1}
		\fivespaces
		{\frac{G[p] \; R \; xy \Lp \Tup{l}{y} \fivespaces l \neq \Lfail}
			{G[p*] \; R \; xy \Lp \Tup{l}{y}}} \mylabel{rep.2}    
		\\ \\ 
		& \tenspaces \tenspaces
		{\frac{G[p] \; R \; xyz \Lp {yz} \fivespaces G[p*] \; R \; yz \Lp z}
			{G[p*] \; R \; xyz \Lp z}} \mylabel{rep.3}
    \\ \\
		& \textbf{Negative Predicate} \;\;\;
		{\frac{G[p] \; \{\} \; xy \Lp \Tup{f}{y}}
			{G[!p] \; R \; xy \Lp xy}} \mylabel{not.1}
		\fivespaces
		{\frac{G[p] \; \{\} \; xy \Lp y}
			{G[!p] \; R \; xy \Lp \Tup{\Lfail}{xy}}} \mylabel{not.2}
			\\ \\
	& \textbf{Recovery} \fivespaces
	{\frac{l \notin Dom(R)}{\Matgk{\throw^{l}}{x}{R} \Lp \Tup{l}{x}}} \mylabel{throw.1}
	\fivespaces
	{\frac{\Matgk{R(l)}{x}{R} \Lp X}
      {\Matgk{\throw^{l}}{x}{R} \Lp X}}  \mylabel{throw.2}
\end{align*}
	}
	\caption{Semantics of PEGs with labeled failures.}
	\label{fig:sempeglab}
\end{figure}

Figure~\ref{fig:sempeglab} presents the semantics of
labeled PEGs with error recovery as a set of inference
rules for a $\Lp$ function. The notation
$\Matgk{p}{xy}{R} \Lp y$ represents a successful matching
of the parsing expression $p$ in the context of a PEG
$G$ against the subject $xy$ with a map $R$ from labels
to recovery expressions, consuming $x$ and leaving the
suffix $y$. By its turn, the notation
$\Matgk{p}{xy}{R} \Lp \Tup{f}{y}$
represents an unsuccessful match of $p$,
where label $f \in L \cup \{ \Lfail \}$,
was thrown when trying to match the suffix $y$.
We will usually use $f$ to represent a label
in the set $L \cup \{ \Lfail \}$, and $l$ to
represent a label in $L$.
The notation $\Matgk{p}{xy}{R} \Lp X$
indicates that the matching result can be
either $y$ or $\Tup{f}{y}$.

The semantics given here is
essentially the same semantics for PEGs with labels
presented in previous work~\cite{medeiros2018visual,medeiros2018sac},
with two simplifications: neither we are tracking
the farthest failure, nor keeping a list of the
errors that occurred during a match. We did this
to make more amenable a formal discussion
about the correct insertion of labels.

We can see in Figure~\ref{fig:sempeglab} that failing
to match a terminal (rules {\bf term.2} and {\bf term.3})
gives us the label \Lfail, while a throw expression
(rules {\bf throw.1} and {\bf throw.2}) may give us a
label different from \Lfail. The recovery map $R$ is simply
passed along. The exceptions are the rules for the syntactic
predicate and for throwing labels.

A label $l \neq \Lfail$ thrown by $\throw^l$ cannot
be caught by an ordered choice or a repetition (rules {\bf ord.2} and {\bf rep.2}),
so it indicates an actual error during parsing, while \Lfail\;
is a regular failure and it indicates that the parser should backtrack. 
In the original formalization of PEGs~\cite{ford2004peg},
there is only label \Lfail, thus the parser always tries
to backtrack after failing to match a parsing expression.

The lookahead operator $!$ captures any label and turns it
into a success (rule {\bf not.1}), while turning a success into
a {\tt fail} label (rule {\bf not.2}). In both rules we 
used an empty recovered map to make sure that
errors are not recovered inside the predicate. 
The rationale is that errors inside a syntactic predicate are
expected and not actually syntactic errors in the input.

Rule {\bf throw.1} is related to error reporting, while
rule {\bf throw.2} is where error recovery happens.
$R(l)$ denotes the recovery expression associated with the label $l$.
When a label $l$ is thrown we check if $R$ has a recovery expression associated with it.
If it does not ({\bf throw.1}), the matching result is $l$ plus the
current input, and this error is propagated, so parsing finishes after
reaching the first syntactical error. 

If label $l$ has a recovery expression $R(l)$ (rule {\bf throw.2}),
we try to match the current input by using $R(l)$. As $R(l)$ is a regular
parsing expression, its matching may succeed, which essentially resumes
regular parsing, or may fail, which may finish the parsing or not
(the parser can still recover from this second error).

When a PEG does not throw labels via expression $\throw^l$, we say it is
an \emph{unlabeled PEG}, as the following definition states:

\begin{definition}[Unlabeled PEG]
A PEG $G = (V,T,P,L,R,\Lfail, p_{S})$ is unlabeled when $\forall A \in V$ we have that
expression $\throw^l$ does not appear in $P(A)$.
\end{definition}

In an unlabeled PEG $G$, the function $R$ is not relevant,
as no label different from $\Lfail$ is thrown and thus rules {\bf throw.1}
and {\bf throw.2} will not be used. In this case, the result of
a matching is more specific, as stated by the following lemma,
where $\Xp$ and $\Xpp$ are suffixes of $x$:

\begin{lemma}
\label{lem:unlabmatch}
Given an unlabeled PEG $G = (V,T,P,L,R,\Lfail, p_{S})$,
$\forall A \in V$, let $p$ be a subexpression of $P(A)$,
either $\Matgk{p}{x}{R} \Lp \Xp$, or $\Matgk{p}{x}{R} \Lp \Tup{\Lfail}{\Xpp}$.
\end{lemma}

\begin{proof}
By induction on the heights of the proof trees for
$\Matgk{p}{x}{R} \Lp \Xp$ and $\Matgk{p}{x}{R} \Lp \Tup{\Lfail}{\Xpp}$.
As $G$ is an unlabeled PEG, we do not have a case related
to expression $\throw^l$, which would involve rules {\bf throw.1}
and {\bf throw.2}. The proof related to the other
expressions is straightforward.
\end{proof}

Below, we discuss an example that illustrates how to deal with syntax
errors in PEGs by using labeled failures and recovery expressions.

\subsection{Handling Syntax Errors in PEGs}

\begin{figure}[t]
{\small
\centering
\begin{align*}
prog        & \leftarrow {\tt PUBLIC\;CLASS\;NAME\;LCUR\;PUBLIC\;STATIC\;VOID\;MAIN}\\
            & \;\;\;\;\;{\tt LPAR\;STRING\;LBRA\;RBRA\;NAME\;RPAR}\;blockStmt\;{\tt RCUR}\;{\tt EOF}\\
blockStmt   & \leftarrow {\tt LCUR}\;(stmt)*\;{\tt RCUR}\\
stmt        & \leftarrow ifStmt\;/\;whileStmt\;/\; printStmt\;/\;decStmt\;/\;assignStmt\;/\;blockStmt\\
ifStmt      & \leftarrow {\tt IF\;LPAR}\;exp\;{\tt RPAR}\;stmt\;({\tt ELSE}\;stmt\;/\;\Epsi)\\
whileStmt   & \leftarrow {\tt WHILE\;LPAR}\;exp\;{\tt RPAR}\;stmt\\
decStmt     & \leftarrow {\tt INT\;NAME}\;({\tt  ASSIGN}\;exp\;/\;\varepsilon)\;{\tt SEMI}\\
assignStmt  & \leftarrow {\tt NAME\;ASSIGN}\;exp\;{\tt SEMI}\\
printStmt   & \leftarrow {\tt PRINTLN\;LPAR}\;{\it exp}\;{\tt RPAR\;SEMI}\\
exp         & \leftarrow relExp \; ({\tt EQ} \; relExp)* \\
relExp      & \leftarrow addExp \; ({\tt LT} \; addExp)* \\
addExp      & \leftarrow mulExp \; (({\tt PLUS} \; / \; {\tt MINUS}) \; mulExp)* \\
mulExp      & \leftarrow atomExp \; (({\tt TIMES} \; / \; {\tt DIV}) \; atomExp)* \\
atomExp     & \leftarrow {\tt LPAR} \; exp \; {\tt RPAR} \; / \; {\tt NUMBER} \; / \; {\tt NAME}
\end{align*}
}
\vspace{-0.7cm}
\caption{A PEG for a tiny subset of Java}
\label{fig:javagrammar}
\end{figure}

In Figure~\ref{fig:javagrammar} we can see a PEG for a tiny subset of Java,
where lexical rules (shown in uppercase) have been elided. 
While simple (this PEG is almost equivalent to an LL(1) CFG),
this subset is a good starting point to discuss error recovery
in the context of PEGs.

To get a parser with error recovery, we first need to have
a parser that correctly reports errors. One popular error
reporting approach for PEGs is to report the farthest
failure position~\cite{ford2002packrat,maidl2016peglabel},
an approach that is supported by PEGs with
labels~\cite{medeiros2018sac}. However, the use of
the farthest failure position makes it harder to recover from
an error, as the error is only known after parsing finishes and
all the parsing context at the moment of the error has been lost.
Because of this, we will focus on using labeled failures for
error reporting in PEGs.

We need to annotate our original PEG with labels,
which indicate the points where we can signal a
syntactical error.
Figure~\ref{fig:javalabels} annotates the PEG of Figure~\ref{fig:javagrammar} (except for the \textit{prog} rule). 
The expression $[p]^{l}$ is syntactic sugar for $(p \; / \; \throw^{l})$. 
It means that if the matching of $p$ fails we should
throw label $l$ to signal an error.

\begin{figure}[t]
{\small
\begin{align*}
prog        & \leftarrow {\tt PUBLIC\;CLASS\;NAME\;LCUR\;PUBLIC\;STATIC\;VOID\;MAIN} \\
            & \;\;\;\;\;\;\;\;{\tt LPAR\;STRING\;LBRA\;RBRA\;NAME\;RPAR}\;\it{blockStmt}\;{\tt RCUR}\;{\tt EOF}\\
blockStmt   & \leftarrow {\tt LCUR}\;(stmt)*\;\labt{RCUR}{rcurlyblk}\\
stmt        & \leftarrow {\it  ifStmt\;/\;whileStmt\;/\; printStmt\;/\;decStmt\;/\;assignStmt\;/\;blockStmt}\\
ifStmt      & \leftarrow {\tt IF}\;\labt{LPAR}{lparif}\;\lab{exp}{condif}\;\labt{RPAR}{rparif}\;\lab{stmt}{thenstmt}\;({\tt ELSE}\;\lab{\it stmt}{elsestmt}\;/\;\Epsi)\\
whileStmt   & \leftarrow {\tt WHILE}\;\labt{LPAR}{lparwhile}\;\lab{exp}{condwhile}\;\labt{RPAR}{rparwhile}\;\lab{stmt}{bodywhile}\\
decStmt     & \leftarrow {\tt INT}\;\labt{NAME}{namedec}\;({\tt ASSIGN}\;\lab{exp}{expdec}\;/\;\varepsilon)\;\labt{SEMI}{semidec}\\
assignStmt  & \leftarrow {\tt NAME}\;\labt{ASSIGN}{assign}\;\lab{exp}{rval}\;\labt{SEMI}{semiassign}\\
printStmt   & \leftarrow {\tt PRINT}\;\labt{LPAR}{lparprint}\;\lab{exp}{expprint}\;\labt{RPAR}{rparprint}\;\labt{SEMI}{semiprint}\\
exp         & \leftarrow relExp \; ({\tt EQ} \; \lab{relExp}{relexp})* \\
relExp      & \leftarrow addExp \; ({\tt LT} \; \lab{addExp}{addexp})* \\
addExp      & \leftarrow mulExp \; (({\tt PLUS} \; / \; {\tt MINUS}) \; \lab{mulExp}{mulexp})* \\
mulExp      & \leftarrow atomExp \; (({\tt TIMES} \; / \; {\tt DIV}) \; \lab{atomExp}{atomexp})* \\
atomExp     & \leftarrow {\tt LPAR} \; \lab{exp}{parexp} \; \labt{RPAR}{rparexp} \; / \; {\tt NUMBER} \; / \; {\tt NAME}
\end{align*}
}
\vspace{-0.7cm}
\caption{A PEG with labels for a small subset of Java}
\label{fig:javalabels}
\end{figure}

The strategy we used to annotate the grammar was to annotate every
symbol (terminal or non-terminal) in the right-hand side of a production
that should not fail, as a failure would just make the whole parser either fail
or not consume the input entirely. For a nearly LL(1) grammar, like the
one in our example, that means all symbols in the right-hand side
of a production, except the first one.
We apply the same strategy when the right-hand side has a choice
or a repetition as a subexpression. 

We can associate each label with an error message.
For example, in rule \textit{whileStmt} the label \texttt{rparwwhile}
is thrown when we fail to match a closing parenthesis, so we could
attach an error message like ``{\tt missing ')' in while}''
to this label. Dynamically, when the matching of {\tt RPAR}
fails and we throw \texttt{rparwhile}, we could enhance this
message with information related to the input position
where this error happened.

Let us consider the example Java program from Figure~\ref{fig:javaerror},
which has two syntax errors: a missing `\textbf{)}' at line 5, and a 
missing `\textbf{;}' at the end of line 7. For this program,
a parser based on the labeled PEG from Figure~\ref{fig:javalabels}
would give us a message like:
\begin{verbatim}
    factorial.java:5: syntax error, missing ')' in while
\end{verbatim}

\begin{figure}[t]
{\small
\begin{verbatim}
    1  public class Example {
    2    public static void main(String[] args) {
    3      int n = 5;
    4      int f = 1;
    5      while(0 < n {
    6        f = f * n;
    7        n = n - 1
    8      }
    9      System.out.println(f);
    10   }
    11 }
\end{verbatim}
}
\caption{A Java program with syntax errors}
\label{fig:javaerror}
\end{figure}

The second error will not be reported because the parser
did not recover from the first one, since {\tt rparwhile} still has
no recovery expression associated with it.

The recovery expression $p_r$ of an label $l$ matches the input from
the point where $l$ was thrown. If $p_r$ succeeds
then regular parsing is resumed as if the label had not been thrown.
Usually $p_r$ should just skip part of the input until is safe to resume parsing.
In rule \textit{whileStmt}, we can see that after the
`\textbf{)}' we expect to match a \textit{stmt}, so the
recovery expression of label \texttt{rparwhile} could skip the
input until it encounters the beginning of a statement.

In order to define a safe input position to resume parsing,
we will use the classical $\fst$ and $\flw$ sets. 
A more detailed discussion about $\fst$ and
$\flw$ sets in the context of PEGs can be found
in other papers~\cite{redz09,redz14,mascarenhas2014}.

With the help of these sets, we can define the following recovery
expression for \texttt{rparwhile}, where
\textit{eatToken} is a rule that matches an input token:
\begin{align*}
(!{\tt \fst(stmt)} \;eatToken)*
\end{align*}

Now, when label \texttt{rparwhile} is thrown,
its recovery expression matches the input until
it finds the beginning of a statement, and
then regular parsing resumes. 

The parser will now also throw label {\tt semiassign}
and report the second error, 
the missing semicolon at the end of line 7.
In case {\tt semiassign} has an associated recovery
expression, this expression will be used to try
to resume regular parsing again.

Even our toy grammar has 26 distinct labels, each needing
a recovery expression to recover from all possible
syntactic errors. While most of these expressions are
trivial to write, this is still burdensome, and for
real grammars the problem is compounded by the fact that
they can easily need a small multiple of this number of labels.
In the next section, we present an approach to automatically
annotate a grammar with labels and recovery expressions
in order to provide a better starting point for larger grammars.

\section{Automatic Insertion of Labels and Recovery Expressions}
\label{sec:algostandard}

The use of labeled failures trades better precision in error messages, and
the possibility of having error recovery, for an increased annotation burden,
as the grammar writer is responsible for annotating the grammar with the
appropriate labels. In this section, we show how this process can be
partially automated.

To automatically annotate a grammar, we need to determine when it is safe
to signal an error: we should only throw a label after expression $p$ fails
if that failure {\em always} implies that the whole parse will fail
or not consume the input entirely, so it is useless to backtrack.

This is easy to determine when we have a nearly $LL(1)$ grammar, as is the case
with the PEG from Figure~\ref{fig:javagrammar}. As we mentioned in
Section~\ref{sec:pegs}, for an $LL(1)$ grammar the general rule is
that we should annotate every symbol (terminal or non-terminal)
in the right-hand side of a production after consuming at least
one token, which in general leads to annotating every symbol 
in the right-hand side of a production except the first one.

Although many PEGs are not $LL(1)$, we can use this approach to
annotate what would be the $LL(1)$ parts of a non-$LL(1)$ grammar.
We will discuss some limitations of this approach in the next section,
when we evaluate its application to annotate PEG-based parsers
for the programming languages Titan, C, Pascal and Java.

While annotating a PEG with labels we can add an automatically
generated recovery expression for each label, based on the tokens
that could follow it. We assume the tokens of a grammar are described
by the non-terminals $A \in \Vlex$. Moreover, we also assume that
at most one non-terminal $A \in \Vlex$ matches a prefix of the
current input, as stated by the following definition:

\begin{definition}[Unique token prefix]
An unlabeled PEG $G \;=\; (\Vlex \cup \Vsyn, T, P, L, R, \Lfail, p_{S})$ has the
unique token prefix property iff $\Matgk{A}{axy}{R} \Lp y$, where $A \in \Vlex$,
then $\forall B \in \Vlex$, where $B \neq A$, we have that $\Matgk{B}{axy}{R} \Lp \Tup{\Lfail}{\Xp}$,
where $\Xp$ is a suffix of $axy$.
\end{definition}

In the above definition, we assumed an unlabeled PEG 
to make sure we would not recover from an error when
matching a lexical non-terminal. Alternatively, we could
have considered above a labeled PEG with an empty recovery
function.

By assuming a grammar with the unique token prefix property
we did not have to worry about which lexical non-terminal
should come first in a choice (e.g., an alternative that
matches \inp{=} can come before one that matches \inp{==}).
Such property is useful, for example, when automatically
computing a choice with the tokens that a recovery expression
should match. The unique token prefix property can be easily
achieved with the help of predicates. For example, we could
define a non-terminal to match input \inp{=} as $ATRIB \leftarrow \tm{=} \; !\tm{=}$.

Moreover, when a PEG $G$ has the unique token prefix property
the sequence of tokens matched for a given input is unique, 
as stated below, where we assumed, as previously, an unlabeled
PEG to avoid recovering in case of an error:

\begin{lemma}[Unique token sequence]
\label{lem:uniseq}
Given an unlabeled PEG $G \;=\; (\Vlex \cup \Vsyn,T,P,L,R,\Lfail, p_{S})$, with
the unique token prefix property, and a subject $w$, the sequence in which the
lexical non-terminals in $\Vlex$ match $w$ is unique.
\end{lemma}

\begin{proof}
By contradiction. Assume the sequence is not unique.
This implies that for some suffix $ax$ of $w$ we would have that
$\Matgk{A}{ax}{R} \Lp \Xp$ and $\Matgk{B}{ax}{R} \Lp \Xpp$,
where $A, B \in \Vlex$, which is not possible given that $G$
has the unique token prefix property. 
\end{proof}

Below, we present Algorithm~\ref{alg:standard},
which automatically adds labels and recovery expressions
to a PEG $G =(\Vlex \cup \Vsyn,T,P,L,R,\Lfail, p_{S})$. 
We assume that all occurrences of $\fst$
and $\flw$ in Algorithm~\ref{alg:standard} give
their results regarding to the grammar $G$ passed
to function $\annotate$. We also assume grammar
$G^\prime$ from function $\annotate$ is available
in function $\addlab$.  

\begin{varalgorithm}{Standard}
    \caption{Adding Labels and Recovery Expressions to a PEG}
    \label{alg:standard}
{\small
\begin{algorithmic}[1]
\Function{\annotate}{$G$}
	\Let{$G^\prime$}{$G$}	
	\For{$A \in \Vsyn$}  
		\Let{$G^\prime(A)$}{$\labexp(G(A), \mathbf{false}, \flw(A))$}
	\EndFor
	\State \Return $G^\prime$
\EndFunction

\State

\Function{\labexp}{$p, seq, flw$}
  \If{$p = A \;\,\mathbf{and}\;\, \Epsi \notin \fst(A) \;\,\mathbf{and}\;\, seq$}
		\State \Return $\addlab(p, flw)$ 
	\ElsIf{$p = p_1\;p_2$}
        \Let{$p_x$}{$\labexp(p_1,seq,\calck(p_2,flw))$}
        \Let{$p_y$}{$\labexp(p_2,seq \;\,\mathbf{or}\;\,  \Epsi \notin \fst(p_1),flw)$}
        \State \Return $p_x \; p_y$
	\ElsIf{$p = p_1 \;/\; p_2$}
        \Let{$p_x$}{$p_1$}
        \If{$\fst(p1) \cap \calck(p2, flw) = \emptyset$}
			\Let{$p_x$}{$\labexp(p_1,\mathbf{false},flw)$}
        \EndIf
        \Let{$p_y$}{$\labexp(p_2,\mathbf{false},flw)$}
		\If{$seq \;\,\mathbf{and}\;\, \Epsi \notin \fst(p_1 \;/\; p_2)$}
			\State \Return $\addlab(p_x \;/\; p_y, flw)$ 
		\Else
			\State \Return $p_x \;/\; p_y$
		\EndIf
	\ElsIf{$p = p_1\!* \;\,\mathbf{and}\;\, \fst(p_1) \cap flw = \emptyset$}
		\State \Return $\labexp(p_1, \mathbf{false}, \fst(p_1) \cup flw)*$
	\Else
		\State \Return $p$
	\EndIf
\EndFunction

\State

\Function{\calck}{$p,flw$}
	\If{$\Epsi \in \fst(p)$}
	    \State \Return ($\fst(p) - \{ \Epsi \}) \cup flw$
	\Else  
		\State \Return $\fst(p)$
	\EndIf
\EndFunction

\State

\Function{\addlab}{$p, flw$}
  \Let{$l$}{newLabel()}
  \Let{$\Rp(l)$}{$(!flw \;eatToken)*$}
	\State \Return $[p]^l$
\EndFunction
\end{algorithmic}
}
\end{varalgorithm}

Function \texttt{\annotate} (lines 1--5) generates a new annotated grammar $G^\prime$
from a grammar $G$. It uses \texttt{\labexp} (lines 7--26) to annotate the
right-hand side, a parsing expression, of each syntactical rule of grammar G.
The auxiliary function \texttt{\calck} (lines 28--32) is used to update the
$\flw$ set associated with a parsing expression. By its turn, the auxiliary
function \texttt{\addlab} (lines 34--37) receives a parsing expression $p$ to annotate
and its associated $\flw$ set $flw$. Function \texttt{\addlab} associates a label $l$ to $p$ and also
builds a recovery expression for $l$ based on $flw$. 
The expression $eatToken$, which matches an input token, can be generated
from the lexical rules of $G$. We assume $G$ has the unique prefix property
when computing $eatToken$ automatically.

Algorithm~\ref{alg:standard} annotates every right-hand side, instead of going
top-down from the root, to not be overly conservative and fail to annotate non-terminals
reachable only from non-LL(1) choices but which themselves might be LL(1).
We will see in Section~\ref{sec:evalstandard} that this has the unfortunate result
of sometimes changing the language being parsed, which is the major
shortcoming of Algorithm~\ref{alg:standard}.

Function \texttt{\labexp} has three parameters. The first one, $p$, is
a parsing expression that we will try to annotate. The second parameter,
$seq$, is a boolean value that indicates whether the current concatenation
consumes at least one token before $p$ or not. 
Finally, the parameter $flw$ represents the $\flw$ set associated with $p$.
Let us now discuss how \texttt{\labexp} tries to annotate $p$.

When $p$ is a non-terminal expression and it
is part of a concatenation that already matched
at least one token (lines 8--9), then we associate a new
label with $p$. In case $p$ represents a non-terminal but $seq$
is not \textbf{true}, we will just return $p$ itself (lines 25--26).
In line 8, we also test whether $A$ matches
the empty string or not. This avoids polluting the grammar
with labels which will never be thrown, since a parsing expression
that matches the empty string does not fail.

In case of a concatenation $p_1\;p_2$ (lines 10--13), we try to annotate
$p_1$ and $p_2$ recursively. To annotate $p_1$ we use an updated $\flw$ set, 
and to annotate $p_2$ we set its parameter $seq$ to \textbf{true} whenever
$seq$ is already \textbf{true} or $p_1$ does not match the empty string.

In case of a choice $p_1\,/\,p_2$ (lines 14--22), we annotate
$p_2$ recursively and in case the choice is disjoint we also
annotate $p_1$ recursively. In both cases, we pass the value
\textbf{false} as the second parameter of \texttt{\labexp},
since failing to match the first symbol of an alternative should
not signal an error. When $seq$ is \textbf{true}, we associate
a label to the whole choice when it does not match the empty
string. 

In case $p$ is a repetition $p_1*$ (lines 23--24), we can annotate $p_1$
if we have a disjoint repetition, i.e., if there is no intersection
between $\fst(p_1)$ and $flw$. 
When annotating $p_1$ we pass \textbf{false} as the second parameter
of \texttt{\labexp} because failing to match the first symbol of a
repetition should not signal an error.

Our concrete implementation of Algorithm~\ref{alg:standard} also
adds labels in case of repetitions of the form $p_1+$, which
should match $p_1$ at least once, and $p_1?$, which should match
$p_1$ at most once. As these cases are similar to the case of $p_1*$,
we will not discuss them here.

Given the PEG from Figure~\ref{fig:javagrammar}, function \texttt{\annotate}
would give us the grammar presented in Figure~\ref{fig:javalabels}
(as previously, we are not taking rule {\it prog} into consideration),
with the exception of the annotation $\lab{stmt}{elsestmt}$.
Label {\tt elsestmt} was not inserted at this point because token
{\tt ELSE} may follow the choice ${{\tt ELSE}\;stmt\;\,/\;\,\Epsi}$,
so this choice is not disjoint (the well-known {\it dangling else} problem).
In Figure~\ref{fig:javalabels}, we associated the label
{\tt elsestmt} to {\it stmt}. This indicates that an {\tt else}
must be associated with the nearby {\it if} statement.

It is trivial to change the algorithm to leave any existing labels
and recovery expressions in place, or to add recovery expressions to
any labels that are already present but do not have recovery expressions.

After applying Algorithm~\ref{alg:standard} to automatically
insert labels, a grammar writer can later add (or remove) labels and
their associated recovery expressions.
We discuss more about this on the next section, where we evaluate
the use of Algorithm~\ref{alg:standard} to add error recovery for
the parsers of several programming languages.

\section{Evaluating Algorithm~\ref{alg:standard}}
\label{sec:evalstandard}

To evaluate Algorithm~\ref{alg:standard}, we built
PEG parsers for the programming languages Titan, C, Pascal
and Java. To build such parsers we used
\texttt{LPegLabel}\footnote{\url{https://github.com/sqmedeiros/lpeglabel}},
a tool that implements the semantics of PEGs with labeled failures, and
\texttt{pegparser}\footnote{\url{https://github.com/sqmedeiros/pegparser}},
which automatically adds labels and recovery expressions to a PEG. 
When building the parsers, we focused on the syntactical
rules, so we have omitted or simplified some lexical rules.

For each language, we first wrote an unlabeled version of the grammar
based on some reference grammar. We have tried to follow the reference
grammar syntactic structure to avoid a bias that could favor our algorithm.
We used a set of syntactically valid and invalid programs to validate
each parser.

Given an unlabeled grammar, we used \texttt{pegparser} to got an
automatically annotated grammar following Algorithm~\ref{alg:standard},
with a recovery expression associated to each label.
We will use the term {\it generated} when referring to
this annotated grammar.

We will compare the generated grammar with a manually annotated
grammar obtained from the unlabeled grammar. We used the same
set of syntactically valid and invalid programs to validate
the generated grammar and the manually annotated one.

In our comparison, we will check the labels of the
generated grammar against the labels of the manually annotated
grammar. We will discuss mainly the following items:
\begin{enumerate}
    \labitem{Equal}{it:equal} When the algorithm correctly inserted a label, as the manual annotation did.
    \labitem{Extra}{it:new} When the algorithm correctly inserted a new label. 
    \labitem{Wrong}{it:wrong} When the algorithm incorrectly inserted a label.
\end{enumerate}

Table~\ref{tab:evallab} shows the result of comparing
the automatically inserted labels with the manually ones.
Below, in Sections~\ref{sec:titanstandard},~\ref{sec:cstandard},~\ref{sec:pascalstandard}
and~\ref{sec:javastandard} we discuss the automatic insertion
of labels for each language.

\begin{table}[t]
{
\small
\begin{subtable}[t]{0.48\textwidth}
\begin{tabular}{|c|c|c|c|}
      \hline
      Approach  & \ref{it:equal} &  \ref{it:new} & \ref{it:wrong} \\ \hline
      Manual    &  86  &  0  &  0  \\ \hline
      Standard  &  76  &  2  &  2  \\ \hline
    \end{tabular}
    \caption{Labels Inserted for Titan}
		\label{tab:evaltitan}
\end{subtable}
\hspace{\fill}
\begin{subtable}[t]{0.48\textwidth}
\flushright
\begin{tabular}{|c|c|c|c|}
      \hline
      Approach  & \ref{it:equal} &  \ref{it:new} & \ref{it:wrong} \\ \hline
      Manual    &  87  &  0  &  0  \\ \hline
      Standard  &  65  &  9  &  1  \\ \hline
    \end{tabular}
    \caption{Labels Inserted for C}
		\label{tab:evalc}
\end{subtable}

\bigskip

\begin{subtable}[t]{0.48\textwidth}
\begin{tabular}{|c|c|c|c|}
      \hline
      Approach  & \ref{it:equal} &  \ref{it:new} & \ref{it:wrong} \\ \hline
      Manual    &  102  &  0  &  0  \\ \hline
      Standard  &  100  &  1  &  3  \\ \hline
    \end{tabular}
    \caption{Labels Inserted for Pascal}
		\label{tab:evalpascal}
\end{subtable}
\hspace{\fill}
\begin{subtable}[t]{0.48\textwidth}
\flushright
\begin{tabular}{|c|c|c|c|}
      \hline
      Approach  & \ref{it:equal} &  \ref{it:new} & \ref{it:wrong} \\ \hline
      Manual    &  175 &  0  &  0  \\ \hline
      Standard  &  139 & 10  & 32  \\ \hline
    \end{tabular}
    \caption{Labels Inserted for Java}
		\label{tab:evaljava}
\end{subtable}
}
\caption{Evaluation of the Labels Inserted by Algorithm~\ref{alg:standard}}
\label{tab:evallab}
\end{table}

Ideally, we would want a generated grammar with the same labels
as the manually annotated one, hopefully with a few new correct
labels missed during manual annotation. To a certain extent, we do
not consider missing to add some labels a serious flaw of Algorithm~\ref{alg:standard},
as long as most of the labels are correctly inserted, since failing to add labels
does not lead to an incorrect parser. These (hopefully few) labels can
still be manually inserted later by an expert.

A discrepancy related to Item~\ref{it:wrong} is more problematic,
since it can produce a parser that does not recognize some
syntactically valid programs. This limitation of our algorithm means that
the output needs to be checked by the parser developer to ensure that the
algorithm did not insert labels incorrectly.

This checking can be done either by manual inspection of the grammar or
by running the generated parser against test programs. In this latter case, when
the parser fails to recognize a valid program, the parsing result will point
the label incorrectly added. Once identified, we need to remove the incorrect
label from the grammar. 

After analyzing how Algorithm~\ref{alg:standard} annotated
the grammar of a given language, we will discuss the error
recovering parser generated by it. 
During this discussion we will assume that we
have already removed the labels that Algorithm~\ref{alg:standard}
may have inserted incorrectly.

As we mentioned, Algorithm~\ref{alg:standard} associates
a recovery expression to each label.
To recover from a label $l$ we add a recovery rule $l$ to
the grammar, where the right-hand side of
$l$ is its recovery expression. 
The generated grammar has a recovery rule associated
with each label.

As \texttt{pegparser} automatically builds an AST
when the match is successful, we will evaluate
the error recovering parser got from a
generated grammar by comparing the AST built by the parser
for a syntactically invalid program with the AST of what
would be an equivalent correct program.
For the AST leaves associated with a syntax error,
we do not require their contents to be the same,
just the general type of the node, so we are comparing
just the structure of the ASTs.

Based on this strategy, a recovery is {\it excellent}
when it gives us an AST equal to the intended one.
A {\it good} recovery gives us a reasonable AST, i.e.,
one that captures most information of the original program
(e.g., it does not miss a whole block of commands).
A {\it poor} recovery, by its turn, produces an AST that
loses too much program information. 
Finally, a recovery is rated as {\it awful}
whenever it gives us an AST without any information
about the program.

Table~\ref{tab:evalrec} shows for how many programs 
of each language the recovery strategy we implemented was
considered {\it excellent}, {\it good}, {\it poor}, or
{\it awful}. 
Sections~\ref{sec:titanstandard},~\ref{sec:cstandard},~\ref{sec:pascalstandard}
and~\ref{sec:javastandard} discuss the results of error
recovery for each language.
In case of the manually annotated grammars, to evaluate them
we added recovery rules based on the way Algorithm~\ref{alg:standard}
generates recovery rules for labels.

\begin{table}[t]
{
\scriptsize
\begin{subtable}[t]{0.48\textwidth}
     \begin{tabular}[t]{|c|c|c|c|c|c|}
      \hline
      Approach  &  Excel.  & Good  & Poor  & Awful   \\ \hline
      Manual    &  91\%    &  4\%   &  5\%  &  0\%   \\ \hline
      Standard  &  81\%    &  3\%   & 15\%  &  1\%   \\ \hline
    \end{tabular}
    \caption{Error Recovery for Titan}
		\label{tab:evalrectitan}
\end{subtable}
\hspace{\fill}
\begin{subtable}[t]{0.48\textwidth}
    \begin{tabular}[t]{|c|c|c|c|c|c|}
      \hline
      Approach &  Excel.  & Good  & Poor   & Awful   \\ \hline
      Manual   &  87\%    &  7\%  &   5\%  &   2\%   \\ \hline
      Standard &  67\%    &  5\%  &  27\%  &   2\%    \\ \hline
    \end{tabular}
    \caption{Error Recovery for C}
		\label{tab:evalrecc}
\end{subtable}

\bigskip

\begin{subtable}[t]{0.48\textwidth}
     \begin{tabular}[t]{|c|c|c|c|c|c|}
      \hline
      Approach  &  Excel.  & Good  & Poor  & Awful   \\ \hline
      Manual    &  80\%    & 11\%   &  9\%  &  0\%   \\ \hline
      Standard  &  80\%    & 11\%   &  9\%  &  0\%   \\ \hline
    \end{tabular}
    \caption{Error Recovery for Pascal}
		\label{tab:evalrecpascal}
\end{subtable}
\hspace{\fill}
\begin{subtable}[t]{0.48\textwidth}
    \begin{tabular}[t]{|c|c|c|c|c|c|}
      \hline
      Approach &  Excel.  & Good   & Poor   & Awful  \\ \hline
      Manual   &  74\%    &  17\%  &   9\%  &   0\%   \\ \hline
      Standard &  63\%    &  16\%  &  21\%  &   0\%    \\ \hline
    \end{tabular}
    \caption{Error Recovery for Java}
		\label{tab:evalrecjava}
\end{subtable}

\caption{Evaluation of Automatic Error Recovery Based on Algorithm~\ref{alg:standard}}
\label{tab:evalrec}
}
\end{table}

To illustrate how we rated a recovery, let us consider
the following syntactically invalid Titan program, where the
range start of the \texttt{for} loop was not given at line 2:
{\small
\begin{verbatim}
    1  sum = 0
    2  for i = , 10 do
    3    print(i)
    4    sum = sum + i
    5  end
\end{verbatim}
}

A recovery would be {\it excellent} in case the AST has all the
information associated with this program (such AST should have
a dummy node to represent the range start). A recovery would
be {\it good} in case the resulting AST misses only the information
about the loop range. By its turn, a recovery would by
rated as {\it poor} in case the resulting AST misses the statements
inside the \texttt{for} (lines 3 and 4).
Lastly, we would rate a recovery as \texttt{awful}
in case it would have produced an AST only with dummy nodes.

Below, based on the approach discussed previously, we evaluate
the use of Algorithm~\ref{alg:standard} to generate error recovering
parsers for the programming languages Titan, C, Pascal and Java.

\subsection{Titan}
\label{sec:titanstandard}

Titan~\cite{titan} is a new statically-typed programming language
under development to be used as a sister language to the Lua programming
language~\cite{lua}.

After some initial development, the Titan parser was manually
annotated with labels to improve its error reporting. 
The original Titan parser~\footnote{\url{http://bit.ly/titan-reference}}
has no error recovery, it stops parsing the input after encountering
the first syntax error. Based on it, we wrote our unlabeled grammar for
Titan~\footnote{\url{http://bit.ly/titan-unlabeled}},
which has 50 syntactical rules.

The Titan grammar is not $LL(1)$, there are non-$LL(1)$ choices
in 7 rules and non-$LL(1)$ repetitions in 3 rules,
but it has many $LL(1)$ parts.

The manually annotated Titan grammar~\footnote{\url{http://bit.ly/titan-manual}}
we got from our unlabeled grammar is equivalent to the original Titan grammar,
we have just adapted the grammar syntax to be able to use the \texttt{pegparser} tool.

The manually annotated grammar has 86 expressions that throw labels.
Some labels, such as {\tt EndFunc}, are thrown more than once, i.e.,
they are associated with more than one expression. 

We then applied Algorithm~\ref{alg:standard} to this unlabeled grammar
and got an automatically annotated Titan grammar,
with a recovery expression associated to each
label~\footnote{\url{http://bit.ly/titan-standard}}. 

In Section~\ref{sec:autolabtitan}, we compare the labels automatically inserted
with the labels in the original Titan grammar. Then, in Section~\ref{sec:autorectitan},
we will discuss the error recovery mechanism of the generated Titan grammar.

\subsubsection{Automatic Insertion of Labels}
\label{sec:autolabtitan}

Algorithm~\ref{alg:standard} annotated the Titan
grammar with 80 labels, which is close to the 86 labels of the
original Titan grammar. A manual inspection revealed that usually the
algorithm inserted labels at the same location of the original ones,
as Table~\ref{tab:evaltitan} shows. We could insert automatically
around 90\% of the labels inserted manually.
Below we discuss the main issues related to the generated
Titan grammar.

As expected our approach did not annotate parts of the 
grammar where the alternatives of a choice were not disjoint,
on in case of a non-disjoint repetition. This happened
in 4 of the 50 grammar rules. One of these rules was \textit{castexp},
which we show below:
\begin{align*}
castexp & \leftarrow simpleexp\; {\tt AS}\; type \;/\; simpleexp
\end{align*}

As we can see, both alternatives of the choice match a {\it simpleexp},
so these alternatives are not disjoint. After manual inspection, we can
see it is possible to add a label to {\it type} in the first alternative, 
since the context where {\it castexp} appears in the rest of the grammar
makes it clear that a failure on {\it type} is always a syntax error.
Left-factoring the right-hand side of {\it castexp} to $simpleexp\;({\tt AS}\;type\;/\;\Epsi)$,
or using the short form $simpleexp\;({\tt AS}\;type)?$,
would give enough context for Algorithm~\ref{alg:standard}
to correctly annotate {\it type} with a label, though.

The manually annotated Titan grammar uses an approach known
as {\it error productions}~\cite{grune2010ptp}.
As an example, the choice associated with rule
{\it statement} has two extra alternatives whose only purpose
it to match some usual syntactically invalid statements, in order
to provide a better error message.
One of these alternatives is as follows:
\begin{align*}
    \&(exp\; {\tt ASSIGN})\; \throw^{{\tt ExpAssign}}
\end{align*}

Before this alternative, the grammar has one that tries to
match an assignment statement. That alternative might have
failed because the programmer used an expression that is not
a valid l-value in the left-hand side of the assignment.
This error production guards against this case. Without the error production, 
the parser would still fail, but we would get an error related to
not closing a function, which may be confusing for a user.

The Algorithm~\ref{alg:standard} does not add error productions,
and we think they should only be added by an expert. 

In case of Titan, the algorithm inserted two labels incorrectly,
a problem related to Item~\ref{it:wrong}, which made the parser
reject valid inputs. Although these two labels have also been added
during the manual annotation, their insertion by Algorithm~\ref{alg:standard}
was undue, as we will see. This issue happened in rules {\it toplevelvar}
and {\it import}. Figure~\ref{fig:titanimport} shows the definition of these
rules, plus some rules that help to add context, in the manually
annotated Titan grammar.

\begin{figure*}
{\small
\begin{align*}
program & \leftarrow  ( toplevelfunc \;/\; toplevelvar \;/\; toplevelrecord \;/\; import \;/\; foreign )* \;\,{\tt EOF}\\
toplevelvar & \leftarrow localopt\; decl\; \labt{ASSIGN}{AssignVar}\; !({\tt IMPORT} \;/\; {\tt FOREIGN})\; \lab{exp}{ExpVarDec}  \\
import      & \leftarrow  {\tt LOCAL}\; \labt{NAME}{NameImport}\; \labt{ASSIGN}{AssignImport} \\
& \;\;\;\;\; \;
                          !{\tt FOREIGN}\; \labt{IMPORT}{ImportImport}\;
                          (\,{\tt LPAR}\; \labt{STRING}{StringLParImport}\; \labt{RPAR}{RParImport} \;/\;\,
                          \labt{STRING}{StringImport}) \\
foreign      &  \leftarrow  {\tt LOCAL}\; \labt{NAME}{NameImport}\; \labt{ASSIGN}{AssignImport} \\
& \;\;\;\;\; \;
                          {\tt FOREIGN}\; \labt{IMPORT}{ImportImport}\;
                          (\,{\tt LPAR}\; \labt{STRING}{StringLParImport}\; \labt{RPAR}{RParImport} \;/\;\,
                          \labt{STRING}{StringImport}) \\
decl & \leftarrow {\tt NAME}\; (\,{\tt COLON}\; \lab{type}{TypeDecl})? \\ 
localopt & \leftarrow {\tt LOCAL}?
\end{align*}
}
\caption{Predicates $!({\tt IMPORT} \;/\; {\tt FOREIGN})$ and $!{\tt FOREIGN}$ enable adding labels after them.}
\label{fig:titanimport}
\end{figure*}

Non-terminals {\it toplevelvar}, {\it import} and {\it foreign} are alternatives of a non-$LL(1)$
choice in rule {\it program}. The parser first tries to recognize {\it toplevelvar},
then {\it import}, and finally {\it foreign}. As a {\it decl} may consist of only a name,
an input like \inp{local x =} may be the beginning of any of these rules.
In rule {\it toplevelvar}, the predicate $!({\tt IMPORT} \;/\; {\tt FOREIGN})$
was added by the Titan developers to make sure the input neither matches the {\it import}
nor the {\it foreign} rule, so it is safe to throw an error after this
predicate in case we do not recognize an expression. The predicate $!{\tt FOREIGN}$
in rule {\it import} plays a similar role. 

As Titan developers inserted these predicates solely to enable
the subsequent label annotations, we judged that we would do a
fairer evaluation by removing them from our unlabeled grammar.

In rule {\it program}, although alternatives {\it toplevelvar}, {\it import},
and {\it foreign} have {\tt LOCAL} in their $\fst$ sets, the algorithm adds
labels to the right-hand side of these non-terminals, because it does not
take into consideration the fact these non-terminals appear as alternatives
in a non-$LL(1)$ choice.

The outcome is that the algorithm is able to insert the same labels added
by manual annotation, but without the syntactic predicates we should not throw
label {\tt AssignImport} in rule {\it toplevelvar} and label {\tt ImportImport}
in rule {\it import}. As Algorithm~\ref{alg:standard} inserted these labels,
the resulting parser will wrongfully signal errors in valid inputs such as
\inp{local x = import "foo"}.

After removing these labels, our generated Titan parser successfully
passed the Titan tests.

We think this was less work than manually annotating the grammar,
given that the parser already needs to have an extensive test suite
that will catch these errors, as was the case in our evaluation.

Lastly, Algorithm~\ref{alg:standard} correctly added two new
labels. It annotated {\tt RARROW} in the first alternative of rule {\it type},
and {\tt FOREIGN} in rule {\it foreign}.

\subsubsection{Automatic Error Recovery}
\label{sec:autorectitan}

The test suite of Titan has 74 tests related to syntactically
invalid programs. For our evaluation of automatic error recovery,
we ran the Titan parser against these files and we analyzed the
AST built for each of them. Since that our parser will only build
an AST for a successful matching, the grammar start rule should
not fail. Thus, as a special case, we should annotate the expressions
of the grammar start rule which may lead to a failure. In case of Titan,
we should annotate {\tt EOF} and add a recovery rule that consumes
the rest of the input. By doing this, we will get an AST whenever
we successfully match an input prefix before matching {\tt EOF}.
We will use this same approach for the other languages.
It is not difficult to extend the
Algorithm~\ref{alg:standard} with this extra case involving
the start rule.

We can see in Table~\ref{tab:evalrectitan}
that our recovery mechanism for Titan seems promising, since that
more than 80\% of the recovery done was considered acceptable,
i.e., it was rated at least {\it good}.

By analysing the programs for which our parser built a poor AST,
we can see that most cases (9 out of 11) are related to missing labels. 
Instead of throwing such labels and recovering
from them using their corresponding recovery expressions, the generated
parser will produce a regular failure, which either leads to the
failure of a matching or makes the parser backtrack.

As an example, let us see the case of a missing label
related to rule {\it castexp}, which we have shown in
Section~\ref{sec:autolabtitan}. In the following input
there is a missing type after the keyword \inp{as} at line 1:
{\small
\begin{verbatim}
    1  x = foo as
    2  return x
\end{verbatim}
}

The manually annotated parser would have thrown an error
after \inp{as}. However, as we have discussed in Section~\ref{sec:autolabtitan},
Algorithm~\ref{alg:standard} did not annotate this rule.
Thus, the automatically generated parser will produce a
regular failure after failing to match {\it type}
after \inp{as}.

This leads the first alternative of rule {\it castexp} to fail,
then the second alternative matches just the input \inp{foo}.
This will lead to another failure when the parser tries to
match \inp{as} as the beginning of a statement.

As Algorithm~\ref{alg:standard} was able to insert most
of the labels inserted by manual annotation, usually the
generated Titan parser was able to recover from an
syntactic error and to build an AST with nearly all the
information about a program.

\subsection{C}
\label{sec:cstandard}

We have developed a parser for C,
without preprocessor directives,
based on the reference grammar
presented by Kernighan and Ritchie~\citep{kernighan1989c},
which is essentially a grammar for ANSI C89.

To write our unlabeled grammar for 
C~\footnote{\url{http://bit.ly/c89-unlabeled}}
we needed to remove left-recursion,
as \texttt{LPegLabel} does not accept grammars with left-recursive
rules. After this, we got an unlabeled grammar for C with
50 syntactical rules, from which 17 have non-$LL(1)$ choices
and 5 have non-$LL(1)$ repetitions.

Due to the \texttt{typedef} feature, to correctly recognize the
C syntax we need the help of semantic actions to determine when
a name should be considered a \textit{typedef\_name}. As we did not
implement these semantic actions, we disabled the matching of
this rule to not incorrectly recognize an identifier as a
\textit{typedef\_name}.

The manually annotated C
grammar~\footnote{\url{http://bit.ly/c89-manual}}
has 87 expressions that throw
labels.
By its turn, the automatically annotated C
grammar~\footnote{\url{http://bit.ly/c89-standard}}
we got after applying Algorithm~\ref{alg:standard}
has 75 labels.

In Section~\ref{sec:autolabc}, we compare the manually annotated
C grammar with the automatically annotated one. After, in Section~\ref{sec:autorecc},
we will discuss the error recovering C parser we got from this automatically annotated
grammar.

\subsubsection{Automatic Insertion of Labels}
\label{sec:autolabc}

As was the case for Titan, often the
Algorithm~\ref{alg:standard} inserted labels at the same location
of the original ones, as Table~\ref{tab:evalc} shows.
The algorithm was able to insert 75\% of the labels inserted manually.

As our C grammar has many rules with non-$LL(1)$ choices
(17 out of 50), and some rules with $non-LL(1)$ repetitions too,
it was not possible to automatically add some labels in these rules.

Algorithm~\ref{alg:standard} incorrectly added one new label,
in rule {\it function\_def}. 
Figure~\ref{fig:cfunctiondef} shows
the definition of this rule, plus other rules that help to add context,
in the generated C grammar.

\begin{figure*}
{\small
\begin{align*}
translation\_unit & \leftarrow  external\_decl+\;\, !{\tt EOF}  \\
external\_decl    & \leftarrow  function\_def  \;\;/\;\;  decl \\
function\_def     & \leftarrow declarator\; decl*\;\, compound\_stat  \;\;/\;\;  decl\_spec\;\, \lab{function\_def}{ErrFuncDef} \\
decl\_spec        &  \leftarrow storage\_class\_spec  \;\;/\;\;  type\_spec \;\;/\;\;  type\_qualifier \\
decl              & \leftarrow  decl\_spec\;\, init\_declarator\_list?\;\, {\tt SEMI} \;\;/\;\;  decl\_spec\;\, \lab{decl}{ErrDecl} 
\end{align*}
}
\caption{Label {\it ErrFuncDef} Incorrectly Added in Rule {\it function\_def}} 
\label{fig:cfunctiondef}
\end{figure*}

The cause of the problem related to Item~\ref{it:wrong}
in the C grammar is similar to the one discussed in Titan
grammar in Section~\ref{sec:autolabtitan}.
In rule {\it external\_decl}, we have a non-$LL(1)$ choice,
since that a {\it decl\_spec} may be the beginning of a
{\it function\_def} as also of a {\it decl}.

When we annotate the right-hand side of the rule associated
with non-terminal {\it function\_def}, which appears in the first
alternative of the non-$LL(1)$ choice in rule {\it external\_decl},
we may throw a label incorrectly. In this case, given an input
like \inp{int x;}, we would match \inp{int} as a {\it decl\_spec}
and we would throw label {\tt ErrFuncDef} after failing to recognize
\inp{x;} as a {\it function\_def}.
After removing label {\tt ErrFuncDef}, our generated C parser
successfully passed the tests.

Finally, Algorithm~\ref{alg:standard} added 9 new labels correctly,
which is more than the 2 new labels added for the Titan grammar.
We think this may be due to the higher rate of non-disjoint expressions
in our C grammar, which may have imposed a more conservative behavior
during manual annotation.

Nevertheless, the manual annotation is not free of faults.
For both grammars some labels were added during manual annotation
and later removed when the parser failed to recognize syntactically
valid programs.

\subsubsection{Automatic Error Recovery}
\label{sec:autorecc}

The test suite we used for our C parser has 59
syntactically invalid programs. As we did for Titan,
we ran the generated C parser against
these files and we analyzed the AST built for each of them.
As we discussed in Section~\ref{sec:autorectitan},
we manually added labels to the grammar start rule
to assure our parser will build an AST when it
successfully matches an input prefix.
In the case of the C grammar, we added two labels
to the right-hand side of the grammar start rule.

In Table~\ref{tab:evalrecc} we can see that
for more than 70\% of the syntactically invalid programs
in our test set the recovery done was considered acceptable,
i.e., it was rated at least {\it good}.

Similarly to Titan (see~\ref{sec:autorectitan}),
in most cases (12 out of 16) we can associate the building of a poor AST
by our parser with the absence of a label. 

As our C grammar has more non-$LL(1)$ choices,
Algorithm~\ref{alg:standard} missed more labels,
which makes a proper recovery more difficult and
results in more poor ASTs.

As an example, let us see the case of a missing label
related to an \texttt{if-else} statement.
Figure~\ref{fig:cifelse} shows the definition of such
statement in rule {\it stat} of the manually
annotated C grammar. Other alternatives of rule {\it stat}
were omitted for simplicity.

\begin{figure*}
{\small
\begin{align*}
stat & \leftarrow    {\tt IF}\; \labt{LPAR}{BrackIf}\; \lab{exp}{InvalidExpr}\; \labt{RPAR}{Brack}\; \lab{stat}{Stat}\; {\tt else}\; \lab{stat}{Stat} \\
     &  \;\; / \;\;  {\tt IF}\; \labt{LPAR}{BrackIf}\; \lab{exp}{InvalidExpr}\; \labt{RPAR}{Brack}\; \lab{stat}{Stat} 
\end{align*}
}
\caption{Manually Annotated {\it if-else} Statement} 
\label{fig:cifelse}
\end{figure*}

As the choice in {\it stat} is not $LL(1)$, 
Algorithm~\ref{alg:standard} will not add the
5 labels to the first alternative of this choice. 
Given a program as the following one, where
there is no statement associated with the {\tt else}:
{\small
\begin{verbatim}
    1  int fat (int x) {
    2    if (x == 0)
    3      return 1;
    4    else
    5  }
\end{verbatim}
}

The generated C parser will try to recognize the first
alternative of the choice in rule {\it stat}. It will
fail to recognize {\it stat} after \inp{else}, which
will produce a regular failure. Thus, the parser backtracks,
recognize an {\it if}-statement without an {\it else}-part,
and then will fail to recognize another statement
as we left \inp{else} on the input.

As we commented out in Section~\ref{sec:autolabtitan},
we could rewrite this choice to put in evidence the common
prefix. After doing this, Algorithm~\ref{alg:standard}
could annotate the {\it if}-statement and we would
get a better recovery in this case.

Although Algorithm~\ref{alg:standard} will not annotate
{\tt LPAR} in the first alternative of the choice above,
this will not make error recovery worst in case of a
missing \inp{(} after \inp{if}, as long as we annotate
{\tt LPAR} in the second alternative. The reason for this
is that after failing to match {\tt LPAR} via the first
alternative, the parser will backtrack and eventually
match {\tt LPAR} via the second alternative.
The same rationale applies for the other labels present
in the common prefix of both alternatives.

\subsection{Pascal}
\label{sec:pascalstandard}

We have developed a parser for Pascal
based on the grammar available in
the ISO 7185:1990 standard~\cite{pascaliso1990}.
Our unlabeled Pascal
grammar~\footnote{\url{http://bit.ly/pascal-unlabeled}}
has 67 syntactical rules. Among these rules, 4 of them
have non-$LL(1)$ choices, and 6 of them have non-$LL(1)$
repetitions.

The manually annotated Pascal
grammar~\footnote{\url{http://bit.ly/pascal-manual}}
has 102 expressions that throw
labels. 

By using Algorithm~\ref{alg:standard},
from the unlabeled Pascal grammar we got
a generated
grammar~\footnote{\url{http://bit.ly/pascal-standard}}
with 104 labels. Below, Section~\ref{sec:autolabpascal} compares the
manually annotated grammar with the generated one, and Section~\ref{sec:autorecpascal} discusses
the error recovering Pascal parser we got from this
generated grammar.

\subsubsection{Automatic Insertion of Labels}
\label{sec:autolabpascal}

As Table~\ref{tab:evalpascal} shows, 
Algorithm~\ref{alg:standard} annotated the Pascal
grammar in a way nearly identical
to manual annotation, it inserted 98\% of the labels
inserted manually. We think the low number of non-$LL(1)$
choices and non-$LL(1)$ repetitions helped the algorithm
to achieve this performance. 

However, three of the labels inserted by
Algorithm~\ref{alg:standard} were added incorrectly.
The incorrect labels were added to
rules {\it subrangeType}, {\it assignStmt} and {\it funcCall}.
All these rules are referenced (directly or indirectly) in
the first alternative of non-$LL(1)$ choices, where an
identifier belong to the $\fst$ set of both choice alternatives.
Let us discuss the problem related to {\it assignStmt},
whose definition is given in Figure~\ref{fig:pascalassign}.

\begin{figure*}
{\small
\begin{align*}
simpleStmt  & \leftarrow  assignStmt  \;\,/\;\,  procStmt  \;\,/\;\,  gotoStmt \\
assignStmt  & \leftarrow  var\; \labt{ASSIGN}{AssignErr}\; \lab{expr}{ExprErr} \\
var         & \leftarrow  {\tt ID}\; ({\tt LBRA}\; \lab{expr}{ExprErr}\; ({\tt COMMA}\; \lab{expr}{ExprErr})*\; \labt{RBRA}{RBrackErr}  \;\,/\;\,  {\tt DOT}\; \labt{ID}{IdErr}  \;\,/\;\,  {\tt POINTER})* \\
procStmt    & \leftarrow  {\tt ID}\; params?
\end{align*}
}
\caption{Label {\it AssignErr} Incorrectly Added in Rule {\it assignStmt}} 
\label{fig:pascalassign}
\end{figure*}

We can see in this figure that there is a non-$LL(1)$ choice
in rule {\it simpleStmt}, as {\tt ID} belongs to the $\fst$ set of
both {\it assignStmt} and {\it procStmt}. Due to this,
in rule {\it assignStmt}, which appears in the first alternative
of this choice, we should not annotate {\tt ASSIGN}, otherwise the
parser will not recognize a valid {\it procStmt} such as \inp{f(x)},
as \inp{:=} does not follow the identifier \inp{f}.

After removing the incorrect labels in rules
{\it subrangeType}, {\it assignStmt} and {\it funcCall},
our generated Pascal parser successfully passed the tests.

Lastly, Algorithm~\ref{alg:standard} also added 2 new labels correctly.

\subsubsection{Automatic Error Recovery}
\label{sec:autorecpascal}

Our test suite for Pascal has 101
syntactically invalid programs. 
We can see in Table~\ref{tab:evalrecpascal} that
for more than 90\% of the syntactically invalid programs
in our test set the recovery done was considered acceptable,
i.e., it was rated at least {\it good}.

Differently from the analysis we did for the Titan
and the C error recovering parsers, in case of the Pascal
parser we can not associate the poor ASTs with the
absence of labels.  A manual
inspection indicates that most of poor ASTs built
were due to synchronizing the input too early
(instead of discarding one more token). This
issue may be fixed by adjusting the recovery expression
used. Our approach allows to do this tuning manually
for a given recovery expression.

Overall, a recovery strategy may show a better performance
after it is tuned to match features of a given language.

\subsection{Java}
\label{sec:javastandard}

We have developed a parser for Java 8
following the parser available at the
Mouse site~\footnote{\url{http://www.romanredz.se/Mouse/Java.1.8.peg}}.

Our unlabeled Java
grammar~\footnote{\url{http://bit.ly/java8-unlabeled}} has
147 syntactical rules, where there are
35 rules with a non-$LL(1)$ choice and
15 rules with a non-$LL(1)$
repetition. A rule may have a non-$LL(1)$ choice
and also a non-$LL(1)$ repetition, but this occurs
in only 2 rules. Overall, one third of the grammar rules
has an $LL(1)$ conflict.
The manually annotated Java
grammar~\footnote{\url{http://bit.ly/java8-manual}}
has 175 expressions that throw
labels. 

From the unlabeled Java grammar, we used
Algorithm~\ref{alg:standard} to get
a generated
grammar~\footnote{\url{http://bit.ly/java8-standard}}
with 181 labels.

In Section~\ref{sec:autolabjava} we compare the
manually annotated grammar with the generated one,
and in Section~\ref{sec:autorecjava} we discuss
our error recovering parser for Java.

\subsubsection{Automatic Insertion of Labels}
\label{sec:autolabjava}

We can see in Table~\ref{tab:evaljava} that
Algorithm~\ref{alg:standard} annotated the Java
grammar with 181 labels, from which 139
were also inserted during the manual annotation.
This seems a good amount, given that many rules
of the grammar have an $LL(1)$ conflict.

The $LL(1)$ conflicts also impose a difficult
to add labels correctly. As a consequence of this,
an important part of the labels added (18\%) by
Algorithm~\ref{alg:standard} were inserted incorrectly.
The cases where these labels were inserted are similar
to the cases of incorrect labels we have already discussed
for the other languages, so we will not present them here.

The significant number of incorrect labels added limits
somewhat the usefulness of using Algorithm~\ref{alg:standard}
to annotate our unlabeled Java grammar, since that it is necessary
to manually remove several labels later. Although
this removal is not hard, the usual process requires
running the tests once for each incorrect label,
and then removing such label after failing to pass
the tests.

Finally, Algorithm~\ref{alg:standard} also correctly added 10 new labels.

\subsubsection{Automatic Error Recovery}
\label{sec:autorecjava}

Our test suite for Java has 175
syntactically invalid programs. 
Table~\ref{tab:evalrecjava} shows that
for almost 80\% of these programs
the recovery done was considered acceptable,
i.e., it was rated at least {\it good}.

About half of the cases where our generated
parser built a poor AST are related to a missing
label. We could get a better
result in these cases by rewriting non-disjoint
choices, as we have shown for Titan
and C, so Algorithm~\ref{alg:standard}
could insert more labels and their corresponding
recovery rules.

For also about half of the cases we got a poor AST
because of an intersection between the tokens that could
follow a symbol in the right-hand side of a rule $A$
and the tokens that could follow $A$ itself. To
improve these ASTs we usually need either
to manually add labels to the grammar or
to manually tune the recovery rules.

\section{Conservative Insertion of Labels}
\label{sec:algocon}

As have discussed previously, Algorithm~\ref{alg:standard}
annotates a grammar with labels, but it may add labels incorrectly,
which leads to a parser that rejects some valid inputs.
To avoid this shortcoming, we will discuss conservative
approaches, which address the problem related to Item~\ref{it:wrong}.

\subsection{Non-Terminals Banning}
\label{sec:ban}

Our first approach to not insert labels incorrectly
is based on the idea of banning a non-terminal $A$
that is used in a non-disjoint choice or a non-disjoint repetition.
When $A$ is banned, we do not
annotate its right-hand side. To properly avoid the wrong
insertion of labels, this approach should be recursive, i.e.,
when banning $A$ we should also ban the non-terminals in
the right-hand side of $A$. 

\begin{figure*}
{\small
\begin{align*}
ordinalType     & \leftarrow  newOrdinalType  \;\;/\;\;  {\tt ID} \\
newOrdinalType  & \leftarrow  enumType  \;\;/\;\;  subrangeType \\
enumType 		    & \leftarrow  {\tt LPAR} \;\, \lab{ids}{IdErr} \;\, \labt{RPAR}{RParErr} \\
subrangeType    & \leftarrow  const\;\, \labt{DOTDOT}{DotDotErr}\;\, \lab{const}{ConstErr} \\
const           & \leftarrow  {\tt SIGN}?\;\, ({\tt UNUMBER}  \;\;/\;\;  {\tt ID})  \;\;/\;\;  {\tt STRING}
\end{align*}
}
\caption{Label {\it DotDotErr} Incorrectly Added in Rule {\it subrangeType}} 
\label{fig:pascalsubrange}
\end{figure*}

To illustrate this point, let us consider Figure~\ref{fig:pascalsubrange},
which shows an excerpt from Pascal grammar. In rule {\it ordinalType} there is
a non-disjoint choice, where {\tt ID} belongs to the $\fst$ set of both
alternatives of the choice. Because of this, we should ban the non-terminal
{\it newOrdinalType}, so we will not annotate its right-hand side.

In case the banning process is not recursive, as in rule {\it newOrdinalType}
there is no conflict, we will not ban the non-terminals in its right-hand side.
This approach leads to incorrectly adding label 
{\tt DotDotErr} in rule {\it subrangeType}.

We should not throw {\tt DotDotErr} because in rule {\it ordinalType},
when matching the first alternative of $newOrdinalType  \;/\;  {\tt ID}$,
the parser could recognize an {\tt ID} as the beginning of a {\it subrangeType},
then fail to recognize {\tt DOTDOT}, backtrack and finally match the second
alternative. Thus, we should apply a recursive banning approach to avoid
adding labels incorrectly.

The result of applying such approach leads to the insertion
of a few labels, or even none. When there are conflicts in the
top-level grammar rules, the recursive banning strategy
bans almost all non-terminals. For the C and Java grammars,
after banning the non-terminals related to a non-disjoint
choice or repetition, we could not add a single label.
In case of Titan, we could add 12 labels, while for Pascal,
which has few non-disjointness conflicts, we had the
best result and could add 36 labels, which corresponds
to 35\% of the labels we have inserted manually.

Although the recursive banning approach have added only
correct labels, its usefulness seems quite limited. Therefore
we will use this strategy only as a complementary one. Below,
we discuss a more effective approach, based on the idea of unique
non-terminals, to conservatively insert only correct labels.

\subsection{Unique Non-Terminals}
\label{sec:unique}

In Section~\ref{sec:algostandard} we saw that
the main challenge when adding labels is to determine statically when
failing to match an expression $p$ indicates that the parser has no other
viable option to recognize the input.

In order to identify these safe places where we can insert
labels, we will introduce the concept of \emph{unique lexical
non-terminals}. The following definition says that a lexical
non-terminal $A$ is \emph{unique} when it appears in the right-hand
side of only one syntactical rule, and just once:
 
\begin{definition}[Unique lexical non-terminal]
Given a PEG $G \;=\; (\Vlex \cup \Vsyn,T,P,L,R,\Lfail, p_{S})$,
$A \in \Vlex$ is unique iff $\exists B \in \Vsyn$ such
that $A$ is used only once in $P(B)$ and $\forall C \in \Vsyn$,
where $C \neq B$, we have $A$ is not used in $P(C)$.
\end{definition}

When we have a grammar $G$ with the unique token prefix property,
and $A$ is a unique lexical non-terminal of $G$, once $A$ matched,
failing to match the expression the follows $A$ leads to the failure
of the whole matching, as the following lemma states:

\begin{lemma}[Unique matching]
\label{lem:unimat}
Let $G \;=\; (\Vlex \cup \Vsyn,T,P,L,R,\Lfail, p_{S})$ be
an unlabeled PEG, with the unique token prefix property,
and let $w$ be a subject $w$.
Let $A \;p_2$ be a subexpression of $P(B)$, where $A$ is a unique lexical
non-terminal and $B \in \Vsyn$, and let $axy$ be a suffix of $w$,
if $\Matgk{A}{axy}{R} \Lp {y}$ and $\Matgk{p_2}{y}{R} \Lp \Tup{\Lfail}{\Yp}$,
then $\Matgk{p_S}{w}{R} \Lp \Tup{\Lfail}{\Wp}$.
\end{lemma}

\begin{proof}
The proof uses Lemma~\ref{lem:uniseq} and the fact that
$G$ has the unique token prefix property and $A$ is a
unique lexical non-terminal.

When the matching of $p_2$ fails, either we backtrack
to a previous choice and try to match a different alternative,
or we do not backtrack.

In the former case, by Lemma~\ref{lem:uniseq} we know that after backtracking
the grammar will match the same sequence of tokens, thus we will need to
match $axy$ again. As $G$ has the unique token prefix property,
only $A$ matches $axy$, and given that $A$ is a unique lexical
non-terminal, $A$ is not used anywhere else in $G$.
Therefore, once more $A$ would match prefix $ax$ and $p_2$ would
fail to match $y$, leading to the failure of the whole matching.

The proof of the last case, when there is no backtracking, is straightforward
given the previous discussion.
\end{proof}

As a result of Lemma~\ref{lem:unimat}, we know that after matching
a unique lexical non-terminal $A$ we start a kind of \emph{unique path},
and failing to match an expression that follows $A$ indicates that the
input is invalid. Therefore, we can safely annotate the expression
$p_2$ that follows $A$.

Based on this, we present the
Algorithm~\ref{alg:unique},
which automatically annotates a PEG
$G =(V,T,P,L,R,{\tt fail}, p_{S})$.
In comparison with the Algorithm~\ref{alg:standard},
function \texttt{\labexp} now receives an extra parameter, $afterU$,
which indicates if we have already matched a unique lexical
non-terminal, and function \texttt{\muni},
which determines whether a parsing expression $p$
matches at least one unique lexical non-terminal or not,
is new. Below we discuss these functions in more detail.
Functions \texttt{\annotate}, \texttt{\calck} and \texttt{\addlab}
remain the same and their definitions were
omitted~\footnote{Actually, now function \texttt{\annotate} provides
a false value to $afterU$ when calling \texttt{\labexp}.}.
We assume the unique lexical non-terminals have already
been computed. Given a non-terminal $A$, function
\texttt{\isunique} returns true in case $A$ is
a unique lexical non-terminal, and false otherwise.

\begin{varalgorithm}{Unique}
    \caption{Inserting Labels and Recovery Expressions in a PEG after Unique Lexical Non-Terminals}
    \label{alg:unique}
{\small
\begin{algorithmic}[1]
%

\Function{\labexp}{$p, seq, afterU, flw$}
  \If{$p = A \;\,\mathbf{and}\;\, \Epsi \notin \fst(A) \;\,\mathbf{and}\;\, seq\;\,\mathbf{and}\;\,afterU$}
		\State \Return $\addlab(p, flw)$ 
	\ElsIf{$p = p_1\;p_2$}
        \Let{$p_x$}{$\labexp(p_1,seq,afterU,\calck(p_2,flw))$}
        \Let{$p_y$}{$\labexp(p_2,seq \;\,\mathbf{or}\;\, \Epsi \notin \fst(p_1),afterU\;\,\mathbf{or}\;\,\muni(p1),flw)$}
        \State \Return $p_x \; p_y$
	\ElsIf{$p = p_1 \;/\; p_2$}
        \Let{$disjoint$}{$\fst(p1) \cap \calck(p2, flw) = \emptyset$}
		\Let{$p_x$}{$\labexp(p_1,\mathbf{false},disjoint\;\,\mathbf{and}\;\,afterU,flw)$}
        \Let{$p_y$}{$\labexp(p_2,\mathbf{false},afterU,flw)$}
		\If{$seq \;\,\mathbf{and}\;\, \Epsi \notin \fst(p_1 \,/\, p_2)\;\,\mathbf{and}\;\, afterU$}
			\State \Return $\addlab(p_x \;/\; p_y, flw)$ 
		\Else
			\State \Return $p_x \;/\; p_y$
		\EndIf
	\ElsIf{$p = p_1\!*$}
	    \Let{$disjoint$}{$\fst(p1) \cap \calck(p2, flw) = \emptyset$}
		\State \Return $\labexp(p_1, \mathbf{false}, disjoint\;\,\mathbf{and}\;\,afterU, \fst(p_1) \cup flw)*$
	\Else
		\State \Return $p$
	\EndIf
\EndFunction

\State

\Function{\muni}{$p$}
	\If{$p = A \;\,\mathbf{and}\;\, \lexrule(A)$}
		\State \Return $\isunique(p)$
	\ElsIf{$p = p_1\;p_2$}
        \State \Return $\muni(p_1) \;\,\mathbf{or} \;\,\muni(p_2)$
	\ElsIf{$p = p_1 \;/\; p_2$}
        \State \Return $\muni(p_1) \;\,\mathbf{and} \;\,\muni(p_2)$
	\ElsIf{$p = p_1\!+$}
		\State \Return $\muni(p_1)$
	\Else
		\State \Return {\bf false}
	\EndIf
\EndFunction

\end{algorithmic}
}
\end{varalgorithm}

Function \texttt{\labexp} (lines 1--20) has four parameters:
$p$, a parsing expression; $seq$, a boolean value indicating 
whether the current concatenation have already matched at least one token;
$afterU$, a boolean value indicating whether the current right-hand side 
have already matched at least one unique lexical non-terminal;
$flw$, the $\flw$ set associated with $p$.

When $p$ is a non-terminal that does not match the empty string
and both $seq$ and $afterU$ are \textbf{true} (lines 2--3), then we associate
a new label with $p$. When $p$ is a non-terminal but these conditions
do not hold, we will just return $p$ itself (lines 19--20).

In case of a concatenation $p_1\;p_2$ (lines 4--7), the main difference
to Algorithm~\ref{alg:standard} is the handling of parameter
$afterU$ when annotating $p_2$ (line 6). In this case, we
supply a \textbf{true} value for $afterU$ when it is already \textbf{true}
or when $p_1$ consumes at least one unique lexical non-terminal.

When $p$ is a choice $p_1\,/\,p_2$ (lines 8--15), a main
difference to Algorithm~\ref{alg:standard} is that we call 
\texttt{\labexp} recursively even when the choice is not disjoint.
In this case, we set $afterU$ to \textbf{false} when annotating
$p_1$ (line 10). The rationale is that is not safe to throw a label after
failing to match $p_1$ in such case, since the parser can still backtrack
and consume the input via $p_2$. We will only add labels to $p_1$ in case
an expression of $p_1$ matches a unique lexical non-terminal. 
When annotating $p_2$, we pass the current value of $afterU$,
since there is no other alternative and thus it is safe to annotate
$p_2$ in case we have matched a unique lexical non-terminal before.
Whether both $seq$ and $afterU$ are \textbf{true},
we associate a label to the whole choice when it does not
match the empty string.

In case $p$ is a repetition $p_1*$ (lines 16--18), 
differently from Algorithm~\ref{alg:standard} and similarly
to the case we discussed before, we also call 
\texttt{\labexp} recursively when the repetition is not disjoint,
providing a false value in this case.

After applying Algorithm~\ref{alg:unique}, we could see we added,
as expected, only correct labels to the grammars we have been
discussing so far. In case of Titan, for example, we added 42
labels, while in case of Java we added 51 labels. To increase
the number of labels inserted, we will do some extra analysis
to determine whether when matching a given expression $p$
we are in a unique path (and thus we can annotate $p$) or not.

Below, we discuss some analyses we did to compute this unique
path. When evaluating Algorithm~\ref{alg:unique}, in Section~\ref{sec:evalcon},
we assume this extra analysis was performed:
\begin{itemize}
  \item \textbf{Unique Syntactical Non-Terminal:} 
  When an syntactical non-terminal $A$ is only used after we have already
  matched a unique lexical non-terminal, then we can also mark $A$ as unique
  and annotate its right-hand side. Both lexical and syntactical non-terminals
  can be marked as unique now, the main difference is that in case of a unique
  syntactical non-terminal this implies providing a true value for parameter
  $afterU$ when calling \texttt{\labexp} to annotate the right-hand side of $A$.

	\item \textbf{Unique Context:}
	If the lexical non-terminal $A$ is used more than once in grammar $G$
  but the set $S$ of tokens that may occur immediately before an usage of $A$
  is unique, i.e., $\forall s \in S$ we have that $s$ may not occur immediately
  before the other usages of $A$, then we can mark this instance of $A$ preceded
  by $S$ as unique.
\end{itemize}

In the next section, we compare the number of labels inserted by Algorithm~\ref{alg:unique}
with the number of labels inserted via manual annotation and by using Algorithm~\ref{alg:standard},
as also as the resulting error recovering parsers obtained via each approach.

\section{Evaluating the Conservative Insertion of Labels}
\label{sec:evalcon}

Table~\ref{tab:evallabuni} shows the amount of labels
inserted for the Titan, C, Pascal and Java grammars
when we used an automatic approach and when we used
manual annotation.

\begin{table}[t]
{
\small
\begin{subtable}[t]{0.48\textwidth}
\begin{tabular}{|c|c|c|c|}
      \hline
      Approach  & \ref{it:equal} &  \ref{it:new} & \ref{it:wrong} \\ \hline
      Manual    &  86  &  0  &  0  \\ \hline
      Standard  &  76  &  2  &  2  \\ \hline
      Unique    &  60  &  3  &  0  \\  \hline
    \end{tabular}
    \caption{Labels Inserted for Titan}
		\label{tab:evalunititan}
\end{subtable}
\hspace{\fill}
\begin{subtable}[t]{0.48\textwidth}
\flushright
\begin{tabular}{|c|c|c|c|}
      \hline
      Approach  & \ref{it:equal} &  \ref{it:new} & \ref{it:wrong} \\ \hline
      Manual    &  87  &  0  &  0  \\ \hline
      Standard  &  65  &  9  &  1  \\ \hline
      Unique    &  48  &  2  &  0  \\  \hline
    \end{tabular}
    \caption{Labels Inserted for C}
		\label{tab:evalunic}
\end{subtable}

\bigskip

\begin{subtable}[t]{0.48\textwidth}
\begin{tabular}{|c|c|c|c|}
      \hline
      Approach  & \ref{it:equal} &  \ref{it:new} & \ref{it:wrong} \\ \hline
      Manual    &  102  &  0  &  0  \\ \hline
      Standard  &  100  &  1  &  3  \\ \hline
      Unique    &   74  &  6  &  0  \\  \hline
    \end{tabular}
    \caption{Labels Inserted for Pascal}
		\label{tab:evalunipascal}
\end{subtable}
\hspace{\fill}
\begin{subtable}[t]{0.48\textwidth}
\flushright
\begin{tabular}{|c|c|c|c|}
      \hline
      Approach  & \ref{it:equal} &  \ref{it:new} & \ref{it:wrong} \\ \hline
      Manual    &  175 &  0  &  0  \\ \hline
      Standard  &  139 & 10  & 32  \\ \hline
      Unique    &   80 & 16  &  0  \\  \hline
    \end{tabular}
    \caption{Labels Inserted for Java}
		\label{tab:evalunijava}
\end{subtable}
}
\caption{Evaluation of the Labels Inserted by the Different Approaches for each Grammar}
\label{tab:evallabuni}
\end{table}

\begin{table}[t]
{
\scriptsize
\begin{subtable}[t]{0.48\textwidth}
     \begin{tabular}[t]{|c|c|c|c|c|c|}
      \hline
      Approach  &  Excel.  & Good  & Poor  & Awful   \\ \hline
      Manual    &  91\%    &  4\%   &  5\%  &  0\%   \\ \hline
      Standard  &  81\%    &  3\%   & 15\%  &  1\%   \\ \hline
      Unique    &  64\%    & 11\%   & 20\%  &  5\%   \\ \hline
    \end{tabular}
    \caption{Error Recovery for Titan}
		\label{tab:evalrecunititan}
\end{subtable}
\hspace{\fill}
\begin{subtable}[t]{0.48\textwidth}
    \begin{tabular}[t]{|c|c|c|c|c|c|}
      \hline
      Approach &  Excel.  & Good  & Poor   & Awful   \\ \hline
      Manual   &  87\%    &  7\%  &   5\%  &   2\%   \\ \hline
      Standard &  67\%    &  5\%  &  27\%  &   2\%    \\ \hline
      Unique   &  55\%    &  3\%  &   3\%  &  38\%    \\ \hline
    \end{tabular}
    \caption{Error Recovery for C}
		\label{tab:evalrecunic}
\end{subtable}

\bigskip

\begin{subtable}[t]{0.48\textwidth}
     \begin{tabular}[t]{|c|c|c|c|c|c|}
      \hline
      Approach  &  Excel.  & Good  & Poor  & Awful   \\ \hline
      Manual    &  80\%    & 11\%   &  9\%  &  1\%   \\ \hline
      Standard  &  80\%    & 11\%   &  9\%  &  1\%   \\ \hline
      Unique    &  61\%    & 15\%   & 24\%  & 11\%   \\ \hline
    \end{tabular}
    \caption{Error Recovery for Pascal}
		\label{tab:evalrecunipascal}
\end{subtable}
\hspace{\fill}
\begin{subtable}[t]{0.48\textwidth}
    \begin{tabular}[t]{|c|c|c|c|c|c|}
      \hline
      Approach &  Excel.  & Good   & Poor   & Awful  \\ \hline
      Manual   &  74\%    &  17\%  &   9\%  &   0\%   \\ \hline
      Standard &  63\%    &  16\%  &  21\%  &   0\%    \\ \hline
      Unique   &  37\%    &   3\%  &  56\%  &   4\%    \\ \hline
    \end{tabular}
    \caption{Error Recovery for Java}
		\label{tab:evalrecunijava}
\end{subtable}

\caption{Evaluating the Error Recovery for Each Approach.}
\label{tab:evalrecuni}
}
\end{table}

We can see that the manual approach is the one
that adds more labels for all the grammars we evaluated,
then comes Algorithm~\ref{alg:standard}, and finally
Algorithm~\ref{alg:unique}, which, as expected, 
did not insert labels incorrectly.

Overall, the amount of labels added by Algorithm~\ref{alg:unique},
when compared with manual annotation, ranged from 55\%,
in case of Java, to 78\%, in case of Pascal.
By its turn, when compared with Algorithm~\ref{alg:standard},
the Algorithm~\ref{alg:unique} was able to insert
between 64\%, in case of Java, and 81\%, in case
of Titan, of the labels inserted by it. 

In Table~\ref{tab:evalrecuni}, we can see that the
smaller amount of labels, and thus of recovery rules,
inserted by Algorithm~\ref{alg:unique} leads to a parser that performs
a poorer recovery when compared to the error recovering parsers based
on manual annotation and on Algorithm~\ref{alg:standard}.
In the best scenario, the Pascal grammar, Algorithm~\ref{alg:unique}
give us a parser that usually (in 84\% of the cases) performs an
acceptable recovery when the other two approaches do. In the worst scenario,
the Java grammar, the parser produced by Algorithm~\ref{alg:unique}
only performs an acceptable recovery in around half of the cases the
other approaches do.

Below, in Sections~\ref{sec:titanuni},~\ref{sec:cuni},~\ref{sec:pascaluni}
and~\ref{sec:javauni}, we discuss in more detail the use of
Algorithm~\ref{alg:unique} to annotate the grammar of each language.

\subsection{Titan}
\label{sec:titanuni}

\begin{figure*}
{\small
\begin{align*}
program     & \leftarrow  ( toplevelfunc \;/\; toplevelvar \;/\; toplevelrecord \;/\; import \;/\; foreign )* \;\,\labt{EOF}{Eof_2}\\
toplevelvar & \leftarrow  localopt\; decl\; \labst{ASSIGN}{AssignVar}\;  \labst{exp}{ExpVarDec}  \\
import      & \leftarrow  {\tt LOCAL}\; \labst{NAME}{NameImport}\; \labst{ASSIGN}{AssignImport} \\
& \;\;\;\;\; \;
                          \labst{IMPORT}{ImportImport}\;
                          (\,{\tt LPAR}\; \labt{STRING}{StringLParImport_3}\; \labt{RPAR}{RParImport_3} \;/\;\,
                          \labt{STRING}{StringImport_3}) \\
foreign      &  \leftarrow  {\tt LOCAL}\; \labt{NAME}{NameImport_2}\; \labt{ASSIGN}{AssignImport_2} \\
& \;\;\;\;\; \;
                          \labt{FOREIGN}{Foreign_2}\; \labt{IMPORT}{ImportImport_1}\;
                          (\,{\tt LPAR}\; \labt{STRING}{StringLParImport_1}\; \labt{RPAR}{RParImport_1} \;/\;\,
                          \labt{STRING}{StringImport_1}) \\
decl & \leftarrow {\tt NAME}\; (\,{\tt COLON}\; \labs{type}{TypeDecl})?  \\
localopt & \leftarrow {\tt LOCAL}?
\end{align*}
}
\caption{Excerpt of Titan Grammar Annotated by Algorithm~\ref{alg:unique}.}
\label{fig:titanimportuni}
\end{figure*}

As we have mentioned in Section~\ref{sec:titanstandard},
the Titan grammar has 7 rules with non-$LL(1)$ choices
and 3 rules with non-$LL(1)$ repetitions.
After applying Algorithm~\ref{alg:unique}, we got
a generated grammar~\footnote{\url{http://bit.ly/titan-unique-scp}}
with 63 labels (around 75\% of the labels added by manual annotation).

Algorithm~\ref{alg:unique} initially identified 44
unique lexical elements in the Titan grammar.
Since that Algorithm~\ref{alg:unique} can annotate the first
alternative of a non-disjoint choice when this alternative has a unique
non-terminal that consumes input, we could add label \texttt{CastMissingType}
in the rule below, where {\tt AS} is a unique lexical non-terminal:
\begin{align*}
castexp & \leftarrow simpleexp\; {\tt AS}\; \labt{type}{CastMissingType} \;/\; simpleexp
\end{align*}

Figure~\ref{fig:titanimportuni} shows an excerpt of
Titan grammar that we discussed in Section~\ref{sec:titanstandard},
without the predicates added by manual annotation.
Non-terminal {\tt FOREIGN} is unique, thus we can annotate the symbols that follow it.
In Figure~\ref{fig:titanimportuni},
we represented as $l_1$ the labels added due to the uniqueness
of {\tt FOREIGN}. 

To get a successful match, the start non-terminal must
succeed, so we mark {\it program} as unique and thus we can annotate
its right-hand side. In case of a repetition $p_1*$, expression
$p_1$ should match at least one token before we can annotate it.
As $p_1$ is a choice, where the same rationale for $p_1*$ applies,
we can not add labels to the alternatives in rule {\it program}.

The syntactical non-terminals {\it toplevelvar}, {\it import}
and {\it foreign} are only used in rule {\it program},
so we could also mark them as unique and annotate
their right-hand side. However, we will only annotate
the right-hand side of {\it foreign}, because it is the
last alternative of the non-disjoint choice in rule
{\it program} involving these non-terminals.
By marking {\it foreign} as unique, we can add
the labels represented as $l_2$ in Figure~\ref{fig:titanimportuni}.

Finally, we can see in Figure~\ref{fig:titanimportuni}
that the lexical non-terminal {\tt IMPORT} is used twice,
so it is not unique. However, each use of {\tt IMPORT}
is preceded by a different context. In rule {\it import},
{\tt ASSIGN} comes before {\tt IMPORT}, while in rule
{\it foreign} it is {\tt FOREIGN} that precedes {\tt IMPORT}.
As we have different contexts, in rule {\it import} we can
add the labels represented as $l_3$. 

The labels represented as $\cancel{lab}$ were added by
Algorithm~\ref{alg:standard} but not by Algorithm~\ref{alg:unique}.
As we discussed in Section~\ref{sec:titanstandard},
labels {\tt ExpVarDec}, in rule {\it toplevelvar},
and {\tt ImportImport}, in rule {\it import}, should
have not been added by Algorithm~\ref{alg:standard}.
As expected, Algorithm~\ref{alg:unique} did not insert
these labels incorrectly.
We should notice that non-terminal {\tt COLON} is used
in other rules of the grammar, which were omitted here,
so it is not unique.

The error recovering parser generated by Algorithm~\ref{alg:unique}
did an acceptable recovery for $75\%$ of the test programs, while
by manually annotating the grammar we could get an acceptable
recovery for $95\%$ of them.

\subsection{C}
\label{sec:cuni}

Our unlabeled C grammar has 50 syntactical rules,
from which 17 have non-$LL(1)$ choices and 5 have
non-$LL(1)$ repetitions. Algorithm~\ref{alg:unique}
was able to generate an error recovering
parser~\footnote{\url{http://bit.ly/c89-uniquescp}}
with 50 labels.
 
In case of our C grammar, Algorithm~\ref{alg:unique}
added only 58\% of the amount of labels inserted manually,
while for Titan it could add 73\% of this amount.
The higher occurrence of non-disjoint expressions in the
C grammar makes more difficult to mark symbols as unique
and also to propagate a unique path after we have seen
a unique symbol. In such grammars, to insert more labels
it seems we need to do a more sophisticated analysis
when computing the unique non-terminals.

Below, we revisit the {\it if-else} statement presented
in Figure~\ref{fig:cifelse} and discuss an extra analyses
we did to mark an usage of {\tt IF} as
unique and helped us to achieve the amount of 50 labels. 
In Figure~\ref{fig:cifelseuni}, we used
$\cancel{lab}$ to represent a label added by manual
annotation but not by Algorithm~\ref{alg:unique}:

\begin{figure*}
{\small
\begin{align*}
stat & \leftarrow    {\tt IF}\; \labst{LPAR}{BrackIf}\; \labs{exp}{InvalidExpr}\; \labst{RPAR}{Brack}\; \labs{stat}{Stat}\; {\tt ELSE}\; \lab{stat}{Stat_1} \\
     &  \;\; / \;\;  {\tt IF}\; \labt{LPAR}{BrackIf_2}\; \lab{exp}{InvalidExpr_2}\; \labt{RPAR}{Brack_2}\; \lab{stat}{Stat_2} 
\end{align*}
}
\caption{{\it if-else} Statement Annotated by Algorithm~\ref{alg:unique}} 
\label{fig:cifelseuni}
\end{figure*}

Initially, the only unique non-terminal is {\tt ELSE}, which allows
us to add just label ${\tt Stat_1}$. Non-terminal {\tt IF} was not
considered unique at first because it is used twice, and both uses
are preceded by the same context. To be able to annotate the second
alternative of a non-disjoint choice such this, we check if the two
usages of a non-terminal $A$ with a context in common occur in the same
right-hand side. If it is the case, we mark the last usage as
unique. After doing this, we could add the labels represented as
$l_2$ in Figure~\ref{fig:cifelseuni}.

We can see in Table~\ref{tab:evalrecunic} that the parser
generated by Algorithm~\ref{alg:unique} performed an acceptable
recovery for 58\% of the test files, while by manually annotating
the grammar we got a 94\% rate of acceptable recovery
for the same test files.

\subsection{Pascal}
\label{sec:pascaluni}

As mentioned in Section~\ref{sec:pascalstandard},
only 10 syntactical rules, out of 67, from the Pascal
grammar have either a non-disjoint choice or a
a non-disjoint repetition. Because of this low
number of non-disjoint expressions, the recursive
banning approach discussed in Section~\ref{sec:ban}
can annotate the Pascal grammar with 36 labels.
By its turn, Algorithm~\ref{alg:unique} was able to
add 72 labels.

As there are eight labels which were only added by
the banning approach, in case of Pascal we automatically
generated an error recovering parser~\footnote{\url{http://bit.ly/pascal-unique-scp}}
which joins the labels added by these two approaches
and thus has 80 labels.

We can see in Table~\ref{tab:evalunipascal} that
Algorithm~\ref{alg:unique} only inserted correct
labels and was able to insert around 80\% of labels
added manually. In Table~\ref{tab:evalrecunipascal} we
can see the resulting error recovering parser
performs an acceptable recovery for 76\% of the
test files, while the parsers based on the other
approaches perform such recovery for 91\% of the
test files.

A manual inspection revealed that the parser generated
by Algorithm~\ref{alg:unique} built an AST with less information
than the parser generated by Algorithm~\ref{alg:standard} 
for the files related to a label inserted only
by Algorithm~\ref{alg:standard}, which shows we got a poorer
recovery in these cases due to the missing labels.

\subsection{Java}
\label{sec:javauni}

In case of our unlabeled Java grammar, where there is a
non-disjoint expression in one third of the 147 grammar rules,
Algorithm~\ref{alg:unique} generated a grammar~\footnote{\url{http://bit.ly/java8-unique-scp}}
with 96 labels. 

As was the case in our C grammar (Section~\ref{sec:cuni}),
the higher occurrence of non-disjoint expressions in
the Java grammar makes more difficult to annotate it.
In case of Algorithm~\ref{alg:standard}, this leaded to
adding 32 labels incorrectly, while in case of the recursive
banning approach discussed in Section~\ref{sec:ban} this
resulted in not adding a single label to the Java grammar.
As Table~\ref{tab:evalunijava} shows,
Algorithm~\ref{alg:unique} was able to add 55\% of
the amount of labels added manually, without
inserting labels incorrectly.

From Table~\ref{tab:evalrecunijava}, we can see that the
error recovering parser generated by Algorithm~\ref{alg:unique}
only performed an acceptable recovery for 40\% of
the test files. This result was somehow expected,
since that the algorithm failed to add many labels that
were inserted during the manual annotation.

\section{Related Work}
\label{sec:rel}

In this section, we discuss some error reporting and recovery approaches
described in the literature or implemented by parser generators.
Overall, a distinctive feature of our approach is that our
error recovery mechanism is integrated with the recognizing
formalism (PEGs, in our case). 

Swierstra~\cite{swierstra2001combinator}
shows a sophisticated implementation of parser combinators for
error recovery. The recovery strategy uses information
about the tails of the pending rules in the parser stack.
When the parser fails to match a given symbol it may insert
this symbol or to remove the current input symbol.

Our approach cannot simulate this recovery strategy, as it relies
on the path that the parser dynamically took to reach the point of
the error, while our recovery expressions are statically determined
from the label. In Swiertra's approach, in case the right-hand side
of the rules are not in some normal form, the parser may have
a high memory consumption.
 
A popular error reporting approach applied for bottom-up parsing
is based on associating an error message to a parse state and a
lookahead token~\cite{jeffery2003lr}. To determine the error
associated to a parse state, it is necessary first to manually
provide a sequence of tokens that lead the parser to that failure state.
We can simulate this technique with the use of labels. By using
labels we do not need to provide a sample invalid program for
each label, but we need to annotate the grammar properly.

The error recovery approach for predictive top-down parsers
proposed by Wirth~\cite{wirth1978algorithms} was a major influence
for several tools. In Wirth's approach, when there
is an error during the matching of a non-terminal $A$, we try to synchronize
by using the symbols that can follow $A$ plus the symbols that can
follow any non-terminal $B$ that we are currently trying to match
(the procedure associated with $B$ is on the stack). Moreover, the
tokens which indicate the beginning of a structured element
(e.g., {\tt while}, {\tt if}) or the beginning of a declaration
(e.g., {\tt var}, {\tt function}) are used to synchronize with
the input.

Our approach can simulate this recovery strategy just partially,
because similarly to~\cite{swierstra1996dec} it relies on information
that will be available only during the parsing. We can define a recovery
expression for a non-terminal $A$ according to Wirth's idea, however,
as we do not know statically how will be the stack when trying to match $A$,
the recovery expression of $A$ would use the $\flw$ sets of all
non-terminals whose right-hand side have $A$, and could possibly
be on the stack.

Coco/R~\cite{cocomanual} is a tool that generates predictive
$LL(k)$ parsers. As the parsers based on Coco/R do not backtrack,
an error is signaled whenever a failure occurs. In case of PEGs,
as a failure may not indicate an error, but the need to backtrack,
in our approach we need to annotate a grammar with labels, a task
we tried to make more automatic.

In Coco/R, in case of an error the parser reports it and
continues until reaching a {\it synchronization point}, which
can be specified in the grammar by the user through the use of a keyword
{\tt SYNC}. Usually, the beginning of a statement or
a semicolon are good synchronization points. 

Another complementary mechanism used by Coco/R for error
recovery is {\emph weak} tokens, which can be defined by
a user though the {\tt WEAK} keyword. A weak token is one that
is often mistyped or missing, as a comma in a parameter list,
which is frequently mistyped as a semicolon. When  the parser fails
to recognize a weak token, it tries to resume parsing
based also on tokens that can follow the weak one.

Labeled failures plus recovery expressions can simulate the
{\tt SYNC} and {\tt WEAK} keywords of Coco/R. 
Each use of {\tt SYNC} keyword would correspond to
a recovery expression that advances the input 
to that point, and this recovery expression would
be used for all labels in the parsing extent of this
synchronization point.
A weak token can have a recovery expression that
tries also to synchronize on its $\flw$ set.

Coco/R avoids spurious error messages during
synchronization by only reporting an error if at least two tokens
have been recognized correctly since the last error. This is easily
done in labeled PEG parsers through a separate post-processing step.

ANTLR~\cite{antlrsite,parr2013antlr} is a popular tool
for generating top-down parsers. 
ANTLR automatically generates from a grammar description a parser
with error reporting and
recovery mechanisms, so the user does not need to annotate
the grammar. After an error, ANTLR parses the
entire input again to determine the error,
which can lead to a poor performance when
compared to our approach~\cite{medeiros2018sac}.

As its default recovery strategy, ANTLR attempts single
token insertion and deletion to synchronize with the input. In case the 
remaining input can not be matched by any production of the 
current non-terminal, the parser consumes the input 
``\textit{until it finds a token that could reasonably follow
the current non-terminal}''~\cite{parr2014antlr}.
ANTLR allows to modify the default error recovery approach,
however, it does not seem to encourage the definition of a 
recovery strategy for a particular error,
the same recovery approach is commonly used for the whole
grammar.

A common way to implement error recovery in PEG parsers
is to add an alternative to a failing expression,
where this new alternative works as a fallback. Semantic actions
are used for logging the error.
This strategy is mentioned in the manual of Mouse~\cite{redzmouse}
and also by users of LPeg~\footnote{See 
\url{http://lua-users.org/lists/lua-l/2008-09/msg00424.html}}.
These fallback expressions with semantic actions for error logging
are similar to our recovery expressions and labels, but in an ad-hoc,
implementation-specific way.

Several PEG implementations such as
Parboiled~\footnote{\url{https://github.com/sirthias/parboiled/wiki}},
Tatsu~\footnote{\url{https://tatsu.readthedocs.io}},
and PEGTL~\footnote{\url{https://github.com/taocpp/PEGTL}} provide 
features that facilitate error recovery.

The previous version of Parboiled used an error recovery
strategy based on ANTLR's one, and requires parsing
the input two or three times in case of an error.
Similar to ANTLR, the strategy used by Parboiled 
was fully automated, and required neither manual intervention
nor annotations in the grammar. Unlike ANTLR, it
was not possible to modify the default error strategy.
The current version of
Parboiled~\footnote{\url{https://github.com/sirthias/parboiled2/}}
does not has an error recovery mechanism.

Tatsu uses the fallback alternative technique for error
recovery, with the addition of a {\it skip expression},
which is a syntactic sugar for defining a pattern that
consumes the input until the skip expression succeeds. 

PEGTL (version 3) makes a distinction between a local failure
(a regular failure), and a global failure, which
is equivalent to throwing a label.	
The PEGTL user can use a function \texttt{raise} to produce
a global failure, which is similar to annotate
a grammar with labels. 

Mizushima et al.~\cite{mizushima2010php} proposed the use of a cut
operator, borrowed from Prolog, to avoid unnecessary
backtracking in PEG parsers, and propose an automatic way to
insert this operator in a grammar. Differently from the
throw operator, which leads to a global failure, in case
there is no recovery rule, the cut operator just discards
the next alternative of the current choice, which makes
difficult the use of cut operators to signal an
error as the parser can still backtrack.
The algorithm proposed by~\cite{mizushima2010php}
to insert the cut operator is similar to 
Algorithm~\ref{alg:standard}. However, the former algorithm
seems to do less insertions, as it does not annotate the second
alternative of a choice, since there is no local backtracking
to discard in this case.

Rüfenacht~\cite{michael2016error} proposes a local
error handling strategy for PEGs. This strategy uses
the farthest failure position and a record of the parser
state to identify an error. Based on the information
about an error, an appropriate recovery set is used.
This set is formed by parsing expressions that match
the input at or after the error location, and it is used
to determine how to repair the input. 

The approach proposed by Rüfenacht is also similar to the use
of a recovery expression after an error, but more limited
in the kind of recovery that it can do. When testing his approach in
the context of a JSON grammar, which is simpler than grammar we
analyzed, Rüfenacht noticed long running test cases and mentions
the need to improve memory use and other performance issues.

The evaluation of our error recovery technique was
based on Pennelo and DeRemmer's~\cite{pennello1978forward} strategy,
which evaluates the quality of an error recovery approach based
on the similarity of the program obtained after recovery with the
intended program (without syntax errors). This quality measure
was used to evaluate several
strategies~\cite{corchuelo2002repair,degano1995comparison,dejonge2012natural},
although it is arguably subjective~\cite{dejonge2012natural}.
 
Differently from  Pennelo and DeRemmer's approach, 
we did not compare programming texts, we compared the AST 
from an erroneous program after recovery with the AST of
what would be an equivalent correct program.

\section{Conclusion}
\label{sec:conc}

We proposed algorithms to automate the process
of adding error reporting and error recovery to parsers
based on Parsing Expression Grammars.
These algorithms annotate a PEG with error labels and associate
recovery expressions for these labels.

We evaluated such algorithms on the grammars of four programming
languages: Titan, C, Pascal and Java. For all these languages,
we build a test suite both for valid and erroneous input.

Algorithm~\ref{alg:standard} could add to these grammars at
least 75\% of the labels added manually. The error recovering
parser we got produced an acceptable recovery for at least 70\%
of the syntactically invalid files of each language.

The major limitation of Algorithm~\ref{alg:standard} is that it can annotate
the right-hand side of a non-terminal $A$ that is used either in
a non-$LL(1)$ choice or in a non-$LL(1)$ repetition. This may
prevent the parser from backtrack and recognize a valid input,
thus changing the grammar language.

To address this issue, we proposed
Algorithm~\ref{alg:unique}, which uses a more
conservative approach, based on the idea of
unique non-terminals. By using it,
we inserted only correct labels and
got an acceptable recovery rate that
ranged from 41\% to 76\%. 

We have also discussed how the rewriting of some 
grammar rules could lead both algorithms to produce
a better result.

The automatic insertion of labels 
provides a good generic error reporting
mechanism. To get more specific error messages,
the parser developer just needs to associate an
error message with each inserted label.

It is easy to adapt our algorithms to use a different
error recovery strategy, which can also be defined
after inserting the labels. It is also possible to
adapt them to work on grammars that have already been
partially annotated, either with just labels or labels and
recovery expressions, as well as marking the parts of the
grammar the algorithm should ignore and that will be
annotated by hand by the parser developer.

To generate a more robust error recovering parsing,
the approach based on unique tokens should insert more
labels. One way to achieve this is by using the derivative
of a PEG~\cite{moss2018derivatives,warth2018derivatives}
to automatically generate valid inputs based on a grammar
without annotations. After applying Algorithm~\ref{alg:unique},
we could repeatedly try to insert a label added only by Algorithm~\ref{alg:standard}
and use the valid input generated through derivatives
to determine whether the insertion of this label is
correct or not.

As a future work, we should also explore other grammar
analysis that may lead Algorithm~\ref{alg:unique} to insert
more correct labels.

Moreover, we may investigate the use of some normal form
when writing a PEG grammar to help our algorithms to produce
a better result, without imposing too much restrictions for
a grammar writer.

Finally, as the use of labeled failures may avoid
unnecessary backtracking, we should also analyze the
performance of the generated parsers.

\section*{References}

\bibliography{scpsblp2018}

\end{document}